%% file: main_file.tex
\documentclass[11pt]{article}
\input{macros}

\makeatletter
\AddToHook{cmd/appendix/before}{%
    \crefalias{section}{appendix}%
    \crefalias{subsection}{appendix}%
}
\makeatother
\allowdisplaybreaks

\setcitestyle{authoryear}

\title{Pricing Pandora's Boxes:\\ Revenue Maximization in Sequential Information Acquisition\footnote{This work was funded in part by NSF award CCF-2217069.}}

\author{
 Shuchi Chawla \\ {\tt shuchi@cs.utexas.edu} \and 
 Dimitris Christou  \\ {\tt christou@cs.utexas.edu} \and
 Trung Dang  \\ {\tt dddtrung@cs.utexas.edu}
 \and
 Zhiyi Huang \\ {\tt zhiyih@cs.utexas.edu}\footnote{All co-authors are affiliated with the University of Texas at Austin.}
}

\date{ }

\begin{document}
\begin{titlepage}
\maketitle
\input{abstract}

\end{titlepage}


\input{introduction}

\input{preliminaries}

\input{mandatory-inspection}

\input{mandatory-extensions}

\input{optional-inspection}

\input{conclusion}

\newpage
\addtocontents{toc}{\protect\setcounter{tocdepth}{1}}
\appendix
\input{appendix}

\newpage
\bibliographystyle{ACM-Reference-Format}
\bibliography{references}

\end{document}

%% file: macros.tex
\usepackage{graphicx} 
\usepackage[utf8]{inputenc}
\usepackage{amsmath}
\usepackage{amsfonts}
\usepackage{caption}
\usepackage{graphicx}
\usepackage{calrsfs}
\usepackage{dirtytalk}
\usepackage{xcolor}
\usepackage{comment}
\DeclareMathAlphabet{\pazocal}{OMS}{zplm}{m}{n}
\DeclareMathAlphabet{\mathcal}{OMS}{cmsy}{m}{n}
\usepackage{enumerate}   
\usepackage{bbm}
\usepackage[hyperfootnotes=false]{hyperref}
\usepackage{amsthm}
\usepackage{thmtools}
\usepackage{cleveref} 
\usepackage{natbib}
\usepackage{comment}
\usepackage{algorithm}
\usepackage{algpseudocode}
\usepackage[margin=1in]{geometry}
\usepackage{pgfplots}
\usepackage{crossreftools}

\pgfplotsset{compat=1.18}

\newtheorem{theorem}{Theorem}
\newtheorem{fact}{Fact}
\newtheorem{claim}[theorem]{Claim}
\newtheorem{lemma}[theorem]{Lemma}

\newtheorem{question}{Question}
\newtheorem{corollary}[theorem]{Corollary}

\newtheorem{inftheorem}{Informal Theorem}

\theoremstyle{definition}
\newtheorem{definition}{Definition}

\crefname{claim}{claim}{claims}
\crefname{inftheorem}{informal theorem}{informal theorems}
\crefname{figure}{figure}{figures}
\crefname{@theorem}{theorem}{theorems}
\crefname{Definition}{definition}{definitions}

\Crefname{claim}{Claim}{Claims}
\Crefname{inftheorem}{Informal Theorem}{Informal Theorems}
\Crefname{@theorem}{Theorem}{Theorems}
\Crefname{Definition}{Definition}{Definitions}
\Crefname{figure}{Figure}{Figures}

\newcommand{\authnote}[2]{[{\color{blue}\textbf{#1:} #2}]}

\newcommand{\dimitris}[1]{\authnote{Dimitris}{#1}}
\newcommand{\trung}[1]{\authnote{Trung}{#1}}

\let\originalleft\left
\let\originalright\right
\renewcommand{\left}{\mathopen{}\mathclose\bgroup\originalleft}
\renewcommand{\right}{\aftergroup\egroup\originalright}

\newcommand{\pbp}{\mathrm{PBP}}
\newcommand{\pb}{\mathrm{PB}}

\newcommand{\pbmdmi}{\mathrm{PBMD}}
\newcommand{\pbmdoi}{\mathrm{PBMDOI}}
\newcommand{\inst}{\mathcal{I}}
\newcommand{\optrev}{\mathrm{OPT}}
\newcommand{\rev}{\mathrm{Rev}}

\newcommand{\dist}{\mathcal{D}}
\newcommand{\mech}{\mathcal{M}}
\newcommand{\vc}{\vec{c}}
\newcommand{\vp}{\vec{p}}
\newcommand{\vq}{\vec{q}}
\newcommand{\vg}{\vec{g}}
\newcommand{\R}{\mathbb{R}}
\newcommand{\N}{\mathbb{N}}
\newcommand{\support}{\operatorname{supp}}
\newcommand{\argmax}{\operatorname{arg\,max}}
\newcommand{\E}{\mathbb{E}}

\newcommand{\expectt}[1]{\mathbb{E}\left[#1\right]}
\newcommand{\probb}[1]{\operatorname{Pr}\left[#1\right]}

\newcommand{\matroid}{\mathrm{M}}
\newcommand{\indep}{\mathcal{I}}
\newcommand{\eps}{\epsilon}
\newcommand{\optlabel}{\mathrm{opt}}
\newcommand{\manlabel}{\mathrm{man}}
\newcommand{\revealset}{\mathcal{R}}
\newcommand{\independent}{\mathbb{I}}
\newcommand{\spn}{\operatorname{span}}

\algrenewcommand\algorithmiccomment[1]{\hfill \(\triangleright\) {\itshape #1}}
\algnewcommand{\LineComment}[1]{\State \(\triangleright\) {\itshape #1}}

%% file: abstract.tex
\begin{abstract}

We study a mechanism design problem in which a seller controls access to information about a set of stochastic alternatives, and a buyer sequentially acquires information in order to choose a single alternative with high value. The value distributions of the alternatives are known to both parties. The seller posts non-adaptive prices for revealing each alternative’s realized value, and the buyer responds optimally by following a Pandora’s Box strategy—deciding which alternatives to inspect and when to stop by accepting the best inspected alternative. The seller's goal is to maximize his expected revenue---the total payment collected from all inspections. 

We study the revenue objective through the lens of simplicity versus optimality. Our main result is that a simple and efficiently computable pricing scheme obtains a $4$-approximation in the worst case to the optimal revenue. This pricing rule equalizes the Weitzman indices across all alternatives. In contrast, we show that equalizing the prices themselves can be an unbounded factor worse than the optimum. Furthermore, for  several natural special cases, including identically distributed alternatives and monotone hazard rate distributions, we fully characterize the optimal pricing.

Finally, we also study a variant of our model under optional inspection, where the buyer may select an alternative without observing its realization. In this setting, we obtain an $n/(n-1)$-approximation for the special case of $n$ identically distributed alternatives, as well as a $2$-approximation for the special case where each alternative’s value distribution has support size two. 

Overall, our results highlight both the computational challenges and the power of simple pricing schemes in selling information to a sequential searcher.

\end{abstract}

%% file: introduction.tex
\section{Introduction}

Consider a buyer searching for a house among several candidates, each with an uncertain true quality. To inform her decision, the buyer may hire a home inspector who, for a fee, can produce a detailed inspection report for any chosen house. Each report reveals valuable information but is costly, so the buyer must decide adaptively which houses to inspect and when to stop and purchase a home. This search problem---balancing the expected benefit of further inspections against their cost---is a classic Pandora’s Box problem and is well understood.

We study this setting from the perspective of a revenue-maximizing home inspector. How should the inspector price inspections to maximize the total payment he receives from the buyer? Charging a high price increases revenue per inspection but may discourage the buyer from inspecting many houses, or even from initiating the search. Should prices differ across houses, and how should they reflect differences in uncertainty or potential value?

These questions fall within the purview of data pricing.
The design of data markets lies at the intersection of mechanism design, information design, and Bayesian persuasion. See, e.g., the survey by \citet{ZBC+24}. 
Prior work in this area typically focuses on {\em what} information to sell and how to {\em package} it---designing complex and correlated signals to influence a receiver’s actions. In contrast, we study a more operational setting in which the information structure is fixed: each ``signal'' corresponds to inspecting a specific alternative and revealing its realized value. The seller does not garble information or design new signals; his only lever is pricing. Moreover, the buyer does not derive intrinsic utility from information itself, but values it solely for its impact on a downstream algorithmic decision, namely which house to purchase.

Our framework captures many practical settings in which the data consumer is an optimizer engaged in sequential information acquisition. For example, a recruiter may purchase background checks or test results for job candidates before making a hire; a firm may commission market research reports for different potential markets before launching a new product. In each case, the decision maker adaptively acquires costly information before selecting a single alternative, while a data provider sets prices in anticipation of this behavior. Our goal is to understand the algorithmic and economic aspects of revenue maximization in such settings.

\paragraph{The Pandora's Box Pricing problem.} Formally, we model the problem as follows. There are $n$ stochastic alternatives $X_1,\dots , X_n$ drawn independently from known priors shared by both the buyer and the seller. The seller publicly announces a price $p_i$ for each alternative $X_i$. The buyer chooses to inspect alternatives in some (adaptive) order; for each inspected alternative $i$, the buyer pays the seller $p_i$ and obtains the instantiation of $X_i$. At any point during the process, the buyer may stop and select the best alternative seen so far, obtaining its realized value. The seller’s revenue is the sum of all prices paid by the buyer throughout this interaction. 

Our model is designed to capture settings in which information is obtained only through costly investigation. In the examples above---home inspections, background checks, and market research---the seller (e.g. the home inspector) does not typically know the realized value of an alternative in advance; rather, he must exert effort to produce the report being sold. We therefore assume that the seller sets prices before observing the realizations of the alternatives,\footnote{We consider a slight generalization in Section~\ref{sec:optional} in which the seller first observes and reveals a subset of the alternatives for free, and then sets prices for the remaining alternatives based on the revealed information.} and that the buyer and seller share the same prior information. Our basic model abstracts away from the seller's cost of producing information, but we show that both the model and our results extend naturally to settings with such costs. Finally, we assume that inspections are sold individually rather than bundled into a single package. This reflects our focus on pricing a fixed menu of information sources, rather than designing or packaging information, and is consistent with applications in which the buyer may wish to explore only a small subset of many available options, such as inspecting a handful of houses on the market or commissioning research on only a few potential market segments.


We allow the seller to set prices
of zero, revealing some information for free, as well as prices of infinity, effectively withholding
certain information altogether. We assume that the buyer is rational and chooses a  strategy that maximizes his expected utility, defined as Utility $:=$ (Value of Selection) $-$ (Prices Paid to the Seller).
We also assume that the buyer can elect to not participate in the mechanism by not selecting any alternative, ensuring that the optimal utility is always non-negative. Therefore, $\expectt{\max_iX_i}$ is an upper bound on the expected seller revenue under any mechanism.




\paragraph{Connections to multi-dimensional revenue maximization.} From the seller's perspective, the pricing problem we study is inherently combinatorial and shares many of the computational challenges of combinatorial item pricing: the seller must choose prices in a multi-dimensional, non-convex space, and the resulting revenue is a non-linear, non-smooth function of these prices. As in standard item-pricing problems, even evaluating the objective at a given price vector is difficult, since it depends on the buyer's optimal response.

At the same time, the structure of the buyer's response differs fundamentally from that in combinatorial item pricing. In our setting, once prices are fixed, the set of ``items'' the buyer purchases---namely, the alternatives she chooses to inspect---is not determined by static linear inequalities over a deterministic value vector. Instead, it emerges endogenously from a stochastic and sequential search process driven by realized values and optimal stopping decisions. The buyer's demand is therefore path-dependent and history-contingent.

As a consequence, standard buyer-utility-based or menu-based characterizations used in combinatorial pricing do not apply, as the buyer's purchases are interleaved with information revelation. 
Moreover, prices influence behavior indirectly by shaping the buyer's exploration incentives, rather than directly screening types via static participation constraints. These features place the seller's problem outside the scope of classical combinatorial pricing models and require new algorithmic and analytical tools.

\subsection{Our results}

To get a feel for the seller's optimization problem, consider the following example with two alternatives $X_1$ and $X_2$ taking either value $1$ or $3$, with $\Pr[X_1 = 1] = 0.2$ and $\Pr[X_2 = 1] = 0.8$. Let $p_1$ and $p_2$ denote the prices set by the seller for the alternatives. We can assume without loss of generality that $p_1\leq \expectt{X_1}=2.6$ and $p_2\leq \expectt{X_2}=1.4$. Given the prices and value distributions, the buyer's optimal strategy, as characterized by \citet{W79}, is to compute an index $g_i$ for each alternative, satisfying $p_i = \expectt{(X_i-g_i)^+}$.\footnote{For a random variable $Y$, $\E[Y^+] := \E[\max\{Y, 0\}]$.} Then, the alternatives are inspected in decreasing order of index until the maximum remaining index becomes either negative or smaller than the best observed realization;\footnote{We assume that the seller has the power to break any ties in his favor.} at this point, the buyer halts and accepts the best realization (if any).

In the current example, we have $g_1=3 - 1.25p_1$ for $p_1 \le 1.6$, and $g_1 = 2.6 - p_1$ otherwise. Calculated similarly, we have $g_2 = 3 - 5p_2$ for $p_2 \le 0.4$, and $g_2 = 1.4 - p_2$ otherwise. Suppose that the prices satisfy $g_1 \ge g_2$, so the alternative $1$ is inspected first. The buyer pays $p_1$ to observe the realization $x_1$ of $X_1$. The second alternative is inspected if and only if $x_1\le g_2$; this event happens with probability $0.2$ if $g_2 \ge 1$ or $p_2 \le 0.4$, and $0$ otherwise. Therefore, the expected revenue of the seller conditioned on $g_1\ge g_2$ is either $p_1 + 0.2 p_2$ if $p_2 \le 0.4$, or $p_1$ otherwise. Optimizing over this space, we get that the best revenue achievable is $2.6$, obtained via $p_1 = 2.6, p_2 = 1.4$, and $g_1 = g_2 = 0$. On the other hand, if we assume the prices satisfy $g_2 \ge g_1$ so that $X_2$ is inspected first, then a similar optimization yields the best revenue of $1.68$, obtained via $p_1 = 1.6, p_2 = 0.4$, and $g_2 = g_1 = 1$. \Cref{fig:rev-example} shows the revenue as a function of $p_1$ and $p_2$.



\begin{figure}[htbp]
    \centering
    \includegraphics{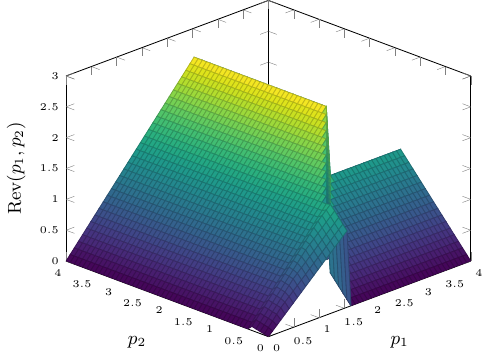}
    \caption{Revenue as a function of $p_1$ and $p_2$}
    \label{fig:rev-example}
\end{figure}

We make a few observations. First, the order in which alternatives are inspected depends in a complex manner on the prices selected, and the best order is not a priori obvious. Furthermore, different orders may yield completely different optimal price profiles. That said, given a desired inspection order, it is possible via dynamic programming to compute the prices that incentivize that order and maximize the seller's revenue.

Second, the maximum possible value the buyer can obtain in expectation is $\E[\max_i X_i]$. The seller cannot, in general, extract this entire value through pricing. In the above example, we have $\E[\max_i X_i]=2.68$. The seller obtains a net expected revenue of $2.6$, the buyer obtains a utility of $0$, and the remaining value of $0.08$ is left on the table because the buyer quits before examining all alternatives.  

Another tantalizing observation is that in the optimal pricing for the above example, the two alternatives have equal indices $g_i$. It is natural to ask whether the optimal pricing always satisfies this property. In fact, for settings with just two alternatives, this is easily seen to be true: for, if not, the price of the first alternative can be raised, lowering its index, but leaving the order and probabilities of inspection for both alternatives unchanged.

Unfortunately this observation does not extend, in general, to settings with more than two boxes. In particular, consider three boxes with indices $g_1=g_2>g_3$ that are inspected by the buyer consecutively. Then lowering $g_1$ and $g_2$ simultaneously increases the prices of the two alternatives without changing the probabilities of inspection of alternatives 1 and 3. However, the probability of inspection for alternative 2 decreases, as the first alternative's instantiation is more likely to exceed the new lowered index of the second. The net effect can be a decrease in the seller's expected revenue. Indeed, we exhibit an explicit example in~\Cref{sec:mandatory} where the optimal indices are non-uniform. 

Nevertheless, our main result shows that setting prices that equalize the indices of all the alternatives (that is, $g_i(p_i)=g$ for all $i\in [n]$) is approximately optimal for the seller. A further observation underlies efficient computation of these indices: we show that the optimal indices are always either $0$ or lie in the union of the supports of the distributions of the values $X_i$. 

Finally, another remarkable property of the above example is that the optimal pricing sets both indices to $0$. While this property also does not hold in general, we show that whenever all of the variables $X_i$ have a monotone hazard rate (MHR), the optimal indices are $0$.

We summarize our main results below.

\begin{inftheorem}
    All instances $(X_1, \dots, X_n)$ of the Pandora's Box Pricing problem admit a $4$-approximation to the seller's revenue via a pricing that equalizes the Weitzman indices of all alternatives (\Cref{thm:4-approx-mandatory-unif-id}). The optimal uniform index lies in the union of the supports of the $X_i$'s and can be computed efficiently (\Cref{cor:efficient}). Furthermore, uniform-index pricing is revenue-optimal for many natural cases, including identically distributed alternatives, monotone hazard rate distributions, and two-point distributions (\Cref{lem:exact-opt-mandatory-special-cases}).
\end{inftheorem}

Given these properties about indices, it is natural to ask whether the same properties hold for the prices themselves: is it possible to achieve a good approximation to revenue by pricing each alternative uniformly? Uniform prices have many desirable properties, from ease of implementation to a perception of fairness and transparency. Unfortunately, we show that in general uniform pricing revenue can be a factor of $\Omega(\log n)$ worse than the optimum (\Cref{thm:uniform-cost-gap}). Moreover, the prices themselves may not lie in the support of the distributions, as the example above illustrates.

\subsection{Extensions} 

Having established our results for the baseline model, we next consider extensions that capture richer settings for pricing information. Let us revisit the running example of a firm evaluating several potential markets for a new product. From the market-research provider’s perspective, producing a report requires effort, time, or money, so the provider may retain only a portion of the posted price after accounting for these costs. From the firm’s perspective, some markets may be entered without commissioning a report, and some may already be partially or fully understood, for instance because of prior experience in related markets. Moreover, the firm may wish to enter multiple markets rather than just one. Each of these considerations changes the incentives on both sides, and consequently affects how the provider should price market-research reports. We study these extensions along the following axes.


\paragraph{Seller Costs and Outside Options.}
We first consider an extension in which each alternative $i$ is associated with a known \emph{seller cost} $c_i \ge 0$ incurred whenever the seller reveals its value, and the buyer has access to a deterministic \emph{outside option} of value $y \ge 0$ that can be chosen without inspection at any time.
As before, the seller posts prices $p_i$ and the buyer responds optimally by solving the induced Pandora's Box problem. Equivalently, the buyer only considers alternatives whose index satisfies $g_i \ge y$, and whenever alternative $i$ is inspected, the seller's net revenue from that inspection is $p_i - c_i$ rather than $p_i$. These features impose natural constraints on the seller: in particular, offering high-cost alternatives may be unprofitable even if they are attractive to the buyer. We show that uniform-index pricings, combined with an appropriate post-processing step, still yield a $4$-approximation to the optimal revenue (\Cref{lem:4-approx-mandatory-unif-id-sellercost}).

\paragraph{Matroid Selection.} Next, we consider the setting where each alternative $i\in [n]$ is associated with an element of an underlying matroid $\matroid([n],\indep)$ (see \Cref{sec:mandatory-ext} for a definition). The buyer sequentially inspects the alternatives like before, but is now allowed to select any independent set of inspected alternatives $S\in\indep$, obtaining their additive value $\sum_{i\in S}X_i$. This variant was first introduced by~\citet{S18}, who showed that the buyer's optimal response is still characterized by Weitzman's indices (although the inspection sequence is modified based on the underlying matroid). We show that for the special case of uniform matroids (corresponding to selecting up to $k$ alternatives for some $k\leq n$), we can efficiently compute uniform-index prices achieving a $4$-approximation to the optimal revenue (\Cref{lem:4-approx-mandatory-unif-id-unifmatroid}). We also show that for general matroids, uniform-index prices perform arbitrarily badly (\Cref{corr:sqrt-logn-gap-partition-matroid-unif-id}) and provide $4$-approximate (non uniform-index) prices for partition matroids (\Cref{corr:4-approx-mandatory-unif-id-partitionmatroid}).

\paragraph{Optional Inspection.}
Finally, we return to the single-selection setting without seller costs, but allow the buyer to select an alternative $i$ \emph{without} inspecting it, in which case they obtain the (random) reward $X_i$ and halt immediately. This variant of the buyer's model is known as \emph{Pandora's Box with Optional Inspection}, introduced by~\citet{D18}.
While this modification may appear modest, it fundamentally changes the structure of the buyer's problem: unlike Weitzman's mandatory-inspection model, computing an optimal utility-maximizing policy for a given price vector is NP-hard, though a PTAS is known~\citep{FLL22,BC23}. Furthermore, the buyer's inspection order is no longer fixed and characterized by some index-based ordering but can be adaptive to past realizations, even when all distributions have support size $2$; we demonstrate such an example in~\Cref{app:proofs-for-optional-examples}.
Consequently, the seller's problem becomes more challenging and the achievable revenue guarantees differ.

We initiate the study of revenue maximization under optional inspection in two directions.
First, for $n$ identically distributed alternatives, we provide a price vector that achieves an $n/(n-1)$-approximation to the optimal revenue (\Cref{thm:optional-iid-approx}).
Second, for distributions with support size $2$, we design a $2$-approximate mechanism in a richer model where the seller may disclose some realizations for free and \textit{then} price the remaining alternatives based on the realizations that both he and the buyer observed (\Cref{thm:optional-bernoulli}).

\vspace{10pt}
In summary, we initiate the study of revenue maximization in sequential information acquisition by introducing the Pandora’s Box Pricing problem and showing that remarkably simple pricing mechanisms can be near-optimal. Our main result gives a constant factor approximation via a pricing that equalizes Weitzman indices across alternatives, and we further show that this uniform-index approach is exactly optimal in several important special cases. We also develop extensions capturing seller costs, outside options, and combinatorial (matroid) selection, and we take the first steps beyond Weitzman’s model by studying optional inspection, where the buyer’s optimal response becomes adaptive and computationally hard; here we obtain constant factor guarantees in special cases, leaving the general setting open.

\subsection{Related Work}
\input{related}

\subsection{Organization}

We formally define our notation and the (base) mandatory inspection model in~\Cref{sec:preliminaries}. All of our results related to this model are discussed in~\Cref{sec:mandatory}. In~\Cref{sec:mandatory-ext}, we formally define the outside option, seller cost and matroid extensions and present all related results. The optional inspection model is discussed in~\Cref{sec:optional}. Finally, we conclude with closing remarks and future directions in~\Cref{sec:conclusion}.

%% file: related.tex
Our work relates to the literature on selling information to imperfectly informed decision makers, which blends information design with pricing. In this line of work, an informed seller designs both the information structure and the terms of trade to extract revenue from a buyer whose actions depend on the information received (see, e.g., \citet{BKP12, BBS18, EsoSzentes2007, HornerSkrzypacz2016, BergemannCaiVelegkasZhao2022, Li2022SellingData, CV11}). A central focus of this literature is the design of signal structures--what information to reveal and how informative or correlated signals should be--subject to Bayes-plausibility constraints. In contrast, we study a setting in which the information structure is fixed and the seller’s only instrument is pricing access to information.

A related but distinct line of work studies the pricing of data and information goods in static or one-shot settings \citep{ZBC+24}. This literature analyzes how a data provider should price datasets, queries, or statistical estimates, often accounting for correlations, arbitrage, or competition among buyers. In contrast to our model, information acquisition in these settings is typically non-adaptive: buyers decide what information to purchase without feedback from realized outcomes. 

Our seller’s problem is also related to the classical question of how to price items for revenue maximization for buyers with stochastic combinatorial preferences. While the optimal pricing is computationally intractable \citep{CDPSY14} and may not approximate the optimal mechanism within any finite factor even in very simple instances \citep{BCKW-JET15, hart2013menu}, small factor approximations are achievable using simple pricing mechanisms under various instance assumptions \citep{CZ17, chawla2015power,chawla2022buy, chawla23buy, CCDHKR24, babaioff2014simple,rubinstein2018simple,cai2019duality}. 


Finally, our work builds on the literature on Pandora’s Box problems and sequential search, originating with the seminal model of \citet{W79}. Variants of this model have been studied extensively, including the optional inspection setting \citep{D18, BK19, FLL22, BC23, DS24}, extensions to correlated values \citep{CGTTZ19, CGMT21, GT24}, constraints on the order of inspection \citep{EHLM19, BFLL20}, combinatorial selection \citep{S18, GJSS19} and more complex models of information acquisition \citep{CCHS24, BLW25, CCD25}. We refer the reader to the survey of \citet{BC24-survey} for a more comprehensive list of recent advances. This literature typically takes inspection costs as exogenous and focuses on characterizing the optimal search policy. In contrast, we endogenize inspection costs by allowing a seller to set prices strategically, anticipating the buyer’s optimal Pandora’s Box strategy. 

To our knowledge, \citet{Arm09}, \citet{Arm17}, and \citet{gamp2022guided} are the only prior works that study endogenously determined prices in the context of sequential search. \cite{Arm09} and \cite{Arm17} assume that access to each alternative $X_i$ is controlled by a distinct seller; different sellers compete to set prices; and the resulting Bertrand equilibrium is analyzed. In contrast, we consider a single information seller who prices inspections across multiple stochastic alternatives. 

\citet{gamp2022guided} studies a setting in which a multi-product monopolist chooses both product prices and inspection costs for multiple items offered to a unit-demand buyer. In their model, inspection costs represent wasteful, time-consuming search borne by the buyer; the seller uses these costs to guide the buyer toward inspecting high-priced products first and ultimately purchasing the most expensive product that yields positive utility. Our setting differs both in what is being sold and in the source of revenue. In \citet{gamp2022guided}, search costs are a means of redirecting demand, while the seller’s revenue comes from product sales. In our model, by contrast, information itself is the product, and the seller’s revenue comes directly from the buyer’s inspection payments. This distinction leads to very different optimization problems. For example, in \citet{gamp2022guided}, it is optimal for the seller to set all indices $g_i$ to zero, whereas in our setting the analogous pricing rule can be far from optimal.


%% file: preliminaries.tex
\section{Preliminaries}\label{sec:preliminaries}

Throughout this paper, we use $[n]:=\{1,2,\dots , n\}$ to denote the set of all alternatives. The values $X_i\ge 0$ of the alternatives are stochastic and are drawn independently from known distributions. 
For each alternative $X_i$, we use $F_{i}(u) = \Pr[X_i \le u]$ to denote its CDF, $f_i(u)$ to denote its PDF, $\support(X_i)$ for its support, and $\mu_i = \expectt{X_i}$ for its expectation.

\paragraph{Prices and Indices.} We write $\vp = (p_1,\dots , p_n)\in\R^n_+$ for a pricing over the alternatives, where $p_i\ge 0$ is the price for inspecting alternative $i$. For each alternative $i$, we use $g_i=g_i(p_i)\in\R$ to denote its Weitzman index under price $p_i$, that is, the unique\footnote{Observe that the index $g_i$ is uniquely defined for all $p_i>0$ and can be any value in $[\max\support(X_i),\infty)$ if $p_i=0$. In the latter case, we adopt the convention that $g_i=\max\support(X_i)$.} solution to  the equation $p_i=\expectt{(X_i-g_i)^+}$. Conversely, we use $p_i(g_i)=\expectt{(X_i-g_i)^+}$ to denote the unique price that induces index $g_i$. Symbolically, we are using $\vp=\vp(\vg)=(p_1(g_1),\dots , p_n(g_n))\in\R^n_+$ and $\vg=\vg(\vp)=(g_1(p_1),\dots , g_n(p_n))\in\R^n$ to denote these transformations. As there is a one-to-one correspondence between prices and indices, we will be using these two parameterizations interchangeably throughout.

\paragraph{The Buyer's Problem.} The buyer is presented with the $n$ alternatives $(X_1,\dots , X_n)$ and a set of prices $\vp\in\R^n_+$. Let $\vg = \vg(\vp)$ be the corresponding indices. The alternatives are inspected in the order given by \textit{any permutation} $\sigma$ such that $g_{\sigma_1}\ge g_{\sigma_2} \ge \dots \ge g_{\sigma_n}$. In particular, let $v_j$ be the best value observed by the buyer after $j$ inspections, for $j\in [n]$, with $v_0=0$. The buyer halts upon inspecting alternative $\sigma_j$ only if $j=n$ or $v_j>g_{\sigma_{j+1}}$, in which case they select the value of the best realized alternative (or a value of $0$, if $j=0$) and gain the corresponding reward. Otherwise, the buyer proceeds to inspect the next alternative $\sigma_{j+1}$ in the order.

\citet{W79} showed that this strategy maximizes the buyer's expected utility, defined as the (expected) reward they eventually collect upon halting \textit{minus} the (expected) sum of prices they pay for inspections. Furthermore, the buyer's optimal utility is always guaranteed to be non-negative, as they are allowed to immediately halt by not inspecting or selecting any of the alternatives. Finally, the optimal utility achieved by the buyer can be derived as follows:
\begin{fact}\label{fact:utility}
     For any prices $\vp$, the buyer's optimal utility is given by $\expectt{\big(\max_{i\in [n]} \min\{X_i,g_i(p_i)\}\big)^+}$.
\end{fact}

\paragraph{The Seller's Problem.}

Given the buyer's optimal response protocol described above, the seller chooses prices $\vp$ to maximize their expected revenue, i.e., the expected sum of prices paid by the buyer. Once $\vp$ is posted, the buyer's behavior (and hence the seller's revenue) is fully determined, except for one important caveat: the buyer may be \emph{indifferent} among multiple inspection orders. In particular, whenever there are ties in the induced indices, the buyer is indifferent among all permutations $\sigma$ that respect the weak ordering of indices.

However, this indifference can matter substantially for the seller. Consider a deterministic alternative $X_1=1+\eps$ and a stochastic alternative $X_2$ taking value $X_2\in\{1,1+\epsilon^{-1}\}$ with probabilities $1-\epsilon$ and $\epsilon$ respectively, for some tiny $\epsilon>0$. The prices $(p_1,p_2) = (\epsilon, 1)$ induce indices $(g_1,g_2)=(1,1)$ and the buyer is indifferent towards the inspection order, always achieving a utility of $1$. If they inspect $X_1$ first, then they observe $1+\epsilon>1$ and they immediately halt; the seller's revenue is $p_1=\epsilon$. On the other hand, if they inspect $X_2$ first, the seller's revenue will be at least $p_2 = 1 \gg \epsilon$.

Since tie-breaking is inconsequential to the buyer yet can substantially affect the seller’s revenue, we adopt the standard modeling convention that ties are broken in the seller’s favor throughout this work. Under this convention, the buyer selects—among all utility-maximizing strategies—the one that maximizes the seller’s expected revenue.\footnote{Alternatively, the seller can also enforce any tie-breaking rule for the buyer up to infinitesimal losses. In particular, to prioritize alternative $1$ over alternative $2$ priced at $(p_1,p_2)$ with $g_1=g_2$, the seller can decrease $p_1$ to $p'_1<p_1$ to set $g'_1=g_1+\epsilon$ for a tiny $\epsilon>0$ small enough to ensure that no distributions contain any point masses in $(g_1, g_1 +\epsilon]$, and no other indices are surpassed. Note that this won't affect the buyer's continuation probabilities and the corresponding loss in revenue will be at most $p_1-p'_1=\expectt{(X_1-g_1)^+-(X_1-g_1-\epsilon)^+}\le \epsilon$.}

We can now define the Pandora's Box Pricing (henceforth, $\pbp$) problem. Whenever an alternative $i$ is inspected, the seller collects its price $p_i=\expectt{(X_i-g_i)^+}$. Furthermore, by the buyer's protocol, this event happens only if $g_i\ge 0$ and all alternatives $j$ that are ordered before $i$ have already been inspected and realized to a value $X_j\le g_i$. Formally:

\begin{definition}[Pandora's Box Pricing]\label{claim:revenue-exp-mandatory}
    An instance $\inst = (X_1,\dots , X_n) $ of $\pbp$ is defined over $n$ stochastic alternatives. For any pricing $\vp\in\R^n_+$ with corresponding indices $\vg = \vg(\vp)$, we say that a permutation $\sigma$ is consistent with $\vp$ if $g_{\sigma_1}\ge\dots \ge g_{\sigma_n}$. The seller's revenue is defined as
    \[\rev(\inst,\vp):=\max_{\sigma:\text{ consistent with $\vp$}}\sum_{i: g_{\sigma_i} \ge 0} \E[(X_{\sigma_i} - g_{\sigma_i})^+] \prod_{j < i} F_{\sigma_j}(g_{\sigma_i}).\]
    We also use $\rev(\inst,\vg) = \rev(\inst,\vg(\vp))$ for the revenue of the pricing that induces indices $\vg$.
\end{definition}

Note that the seller’s revenue is fully determined by the indices induced by the posted prices. To avoid repetition, from now on, whenever we state an ordering $g_{\sigma_1}\ge \dots \ge g_{\sigma_n}$ of the alternatives with respect to some pricing $\vp$, we will implicitly assume $\sigma$ is the optimal consistent permutation of $\vp$ and therefore that the buyer inspects the alternatives in that order. Furthermore, we will be implicitly assuming that $g_i\ge 0$ for any alternative $i$; otherwise, we can decrease its price from $p_i$ to $\mu_i$ so that $g_i(\mu_i)=0$; as the buyer would never pay price $p_i$, this transformation never harms the seller's revenue.

Finally, a central part of this work is devoted to studying a particular family of pricing schemes that equalize the indices across all alternatives. Formally:

\begin{definition}[Uniform-Index Prices]
    For any $\pbp$ instance $\inst = (X_1, \dots, X_n)$ and a fixed index $g \in \R$, we let $\vp(g)$ be the unique price vector that induces index $g$ across all alternatives, i.e. $p_i = \E[(X_i - g)^+]$ for all $i \in [n]$. The family of uniform-index prices is then $\{\vp(g) : g \in \R\}$ and we write $\rev(\inst,g):=\rev(\inst,\vp(g))$ to denote the corresponding revenues.
\end{definition}

\noindent The objective of $\pbp$ is to determine the pricing $\vp$ that achieves the maximum revenue in a given instance $\inst$, which we denote as $\optrev(\inst):=\max_{\vp\in\R^n_+} \rev(\inst,\vp)$\footnote{From~\Cref{claim:revenue-exp-mandatory}, $\mathrm{Rev}(\vec{g})$ is an upper semicontinuous function and $\mathrm{Rev}(\vec{g})\in [0,\sum_{i=1}^n\mathbb{E}[X_i]]$ for all $\vec{g}\in \mathbb{R}^n_+$. We further assume all expectations to be finite as otherwise we can achieve infinite revenue by pricing just one alternative; therefore, we can assume that the revenue function is bounded. By definition of the indices we have that $g_i\in [0,\sup_i\mathrm{support}(X_i)]$ and therefore for distributions with bounded supports, a maximizer of $\mathrm{Rev}(\vec{g})$ always exists. For unbounded supports, since $\mathrm{Rev}(\vec{g})$ is bounded, a supremum exists and therefore there exists a set of indices $(g'_1,\ldots , g'_n)$ that $\epsilon$-approximate this supremum up to an (arbitrarily small) additive error $\epsilon>0$. All of our proofs can then be stated with respect to $\vec{g}'$.}. 
When clear from the context, we omit the dependency on the instance $\inst$ and write $\optrev$ and $\rev(\vp)$, $\rev(\vg)$, and $\rev(g)$.

%% file: mandatory-inspection.tex
\section{Approximately Optimal Pricings for Pandora's Box} \label{sec:mandatory}
In this section, we study the seller's problem of maximizing their revenue under a given $\pbp$ instance $\inst = (X_1, \dots, X_n)$.

Recall the two-alternative example discussed in the introduction. As we saw, the optimal revenue was obtained by setting prices $(p_1,p_2)=(\mu_1,\mu_2)$ or equivalently, uniform indices $(g_1,g_2)=(0,0)$. Since setting any price $p_i>\mu_i$ ensures that $g_i<0$ and the buyer never inspects the alternative, one may ask whether ``greedily'' maximizing prices by setting $g_i=0$ for all $i$ is always near-optimal. We demonstrate that this is not the case.

Consider $n$ identically distributed alternatives taking either value $\epsilon^{2}$ or $1$ with $\Pr[X_1 = 1] = \epsilon$ for some tiny $\epsilon>0$. If we set their prices to $p_i=\mu_i =\epsilon^2(1-\epsilon) + \epsilon$, then the buyer will inspect a single alternative, realize a value $X_1>0$ and halt; the seller's revenue will therefore be $\mu_1\approx \epsilon$. Now consider a different pricing $\vp=\vp(\epsilon^2)$, corresponding to (uniform) prices $p = \epsilon(1-\epsilon^2)$ and (uniform) indices $g_i=\epsilon^2$. The buyer will start inspecting the alternatives until one of them gets realized to $X_i=1$ and then stop. Therefore, the expected revenue of this pricing will be $p + p(1-\epsilon) + \dots + p(1-\epsilon)^{n-1} = (1-\epsilon^2)\cdot(1-(1-\epsilon)^n)\approx n\epsilon$ for sufficiently small $\epsilon$. 

In other words, the seller must carefully balance maximizing prices against maximizing the buyer's exploration. Nevertheless, we show that in the special case of MHR distributions, setting $g=0$ is indeed optimal. Furthermore, for many other special cases, uniform-index prices $\vp(g)$ for some $g\in \R$ (not necessarily $0$) suffice to achieve the optimal revenue. 

\begin{theorem} \label{lem:exact-opt-mandatory-special-cases}
    Let $\inst$ be any $\pbp$ instance satisfying (at least) one of the following:
    \begin{enumerate}
        \item There are only $n=2$ alternatives.

        \item All alternative distributions are identical.

        \item All alternative distributions have support size $\leq 2$.
        
        \item All alternative distributions possess a monotone hazard rate.\footnote{A random variable $Y$ has monotone hazard rate if $h_Y(x):=f_Y(x)/(1-F_Y(x))$ is a non-decreasing function.}
    \end{enumerate}
    Then, there exist uniform-index prices $\vp(g)$ such that $\rev(\inst,\vp(g))=\optrev(\inst)$. In Case (4) the optimal uniform index is given by $g=0$.
\end{theorem}
\begin{proof}
    Cases (1) and (2) are proven in~\Cref{appendix:mandatory-special-proofs}. Here, we prove cases (3) and (4).

    \paragraph{Two Point Distributions.}
    We begin with case (3) and use $\ell_i \le r_i$ to denote the two values in the support of $X_i$. Let $\vp$ be any optimal pricing and $\vg = (g_1,\dots , g_n)$ be the corresponding indices with $g_1\ge \dots \ge g_n\ge 0$. We assume that $g_i<r_i$ for all $i\in [n]$, as otherwise $p_i=0$ and it doesn't contribute to the revenue; we can instead set $p_i=\mu_i$ so that $g_i=0$ and push alternative $i$ to the end of the buyer's priority without affecting the inspection probabilities of the remaining alternatives. 
    
    Let $k$ be the maximum label with $g_k \ge \max_{j < k} \ell_j$. Observe that as $g_{k + 1} < \max_{j \le k} \ell_j$, the buyer never inspects any alternative with label $k + 1$ or larger. We claim that setting uniform prices $\vp_u=\vp(g_k)$ (and maintaining the same inspection order) weakly increases the seller's revenue:

    \begin{itemize}
        \item All alternatives $i\in \{k+1,\dots , n\}$ were never inspected under $\vg$ and therefore didn't contribute to the revenue. Furthermore, they are ordered after any alternative $j\le k$ in both pricings and therefore they do not affect the probability of inspecting $j$.

        \item All alternatives $i\in \{1,\dots , k\}$ are now priced higher, since their index decreases from $g_i$ to $g_k$ and therefore $\expectt{(X_i-g_i)^+}\le \expectt{(X_i-g_k)^+}$. Therefore, it suffices to argue that the probabilities of inspecting them do not decrease. Initially, alternative $i\le k$ was inspected with probability $\prod_{j<i}F_j(g_i)$ whereas now it is inspected with probability $\prod_{j<i}F_j(g_k)$; we will argue that $F_j(g_i)=F_j(g_k)$ for all $j<i\le k$. Indeed, we have $\ell_j \le g_k \le g_i \le g_j < r_j$ by definition of $k$ and therefore $F_j(g_i)=F_j(g_k) = \probb{X_j=\ell_j}$.
    \end{itemize}

    \paragraph{MHR Distributions.}
    We now focus on case (4). It is a well-known fact that if a random variable $X$ is MHR, then the mean residual life (that is, the quantity $\E[X - g \mid X > g]$) is non-increasing in $g$~\citep{BMP63}.\footnote{In particular,
    $\E[X - g \mid X > g] = \int_0^\infty \exp\left(-\int_{g}^{g+t} h_X(y) \, dy\right) \, dt$ and the proof follows by the monotonicity of $h_X(\cdot)$.} Furthermore, as MHR variables are continuous, $\E[X \mid X > 0] = \E[X]$. 

    Consider the prices $\vp_0 = \vp(0) = (\mu_1, \dots , \mu_n)$ corresponding to $g_i=0$ for all $i$. The buyer is indifferent towards which alternative to inspect first; and once an alternative is inspected, the buyer will immediately accept it as $\probb{X_i>0}=1$. Therefore, the seller's revenue under $\vp$ is maximized if the buyer prioritized alternative $j\in \argmax_i \mu_i$, and we obtain $\rev(\vp_0)=\max_i\mu_i$. To complete the proof, we will argue that $\optrev\le \max_i\mu_i$. Let $\vp$ be any revenue-maximizing set of prices with $\vg(\vp) = (g_1,\dots , g_n)$ and re-label the alternatives via the buyer's inspection order so that $g_1\ge \dots \ge g_n\ge 0$; alternatives with $g_i<0$ can be ignored. We have:

    \begin{align*}
        \optrev &= \sum_{i=1}^n \prod_{j=1}^{i-1} F_{j}(g_i) \cdot \E[(X_i - g_i)^+] \\
            &\le \sum_{i=1}^n \left(\prod_{j=1}^{i-1} F_{j}(g_j)\right) \Pr[X_i > g_i] \cdot \E[X_i - g_i \mid X_i > g_i]  && \text{($g_j\geq g_i$ for all $j<i$)}\\
            &\le \sum_{i=1}^n \left(\prod_{j=1}^{i-1} F_{j}(g_j)\right) \Pr[X_i > g_i] \cdot \E[X_i \mid X_i > 0] && \text{($g_i\ge 0$)}\\
            &\le \sum_{i=1}^n \left(\prod_{j=1}^{i-1} \Pr[X_j \le g_j]\right) \Pr[X_i > g_i] \cdot \left(\max_k \mu_k\right) && \text{($\E[X_i \mid X_i > 0]=\mu_i$)}\\
            &\le \max_k \mu_k
    \end{align*}
    with the last inequality following from the fact that $q_i :=  \Pr[X_i > g_i]\cdot\prod_{j=1}^{i-1} \Pr[X_j \le g_j] $ is the probability of the event that ``alternative $i$ is the first alternative whose value exceeds its index'' and therefore $\sum_i q_i \le 1$.
\end{proof}

\Cref{lem:exact-opt-mandatory-special-cases} raises the question of whether uniform-index pricings could be optimal in general. 
We answer this question in the negative by proving the following (constant) lower bound on the approximate optimality of uniform-index prices. We defer the proof to~\Cref{app:unif-index-lb}.

\begin{theorem} \label{thm:const-gap-mandatory-unif-id}
    There exists a $\pbp$ instance $\inst$ with $\optrev(\inst) \geq 1.14\cdot\max_{g \in \R}\rev(\inst,\vp(g))$.
\end{theorem}

On the positive side, our main result for this section shows that while sub-optimal, uniform-index prices always obtain a constant approximation to the optimal revenue, for all instances.

\begin{theorem} \label{thm:4-approx-mandatory-unif-id}
     For all $\pbp$ instances $\inst$, we have $\optrev(\inst)\leq 4\cdot \max_{g \in \R}\rev(\inst,\vp(g))$.
\end{theorem}
\begin{proof}

Let $\vp$ be any revenue-maximizing set of prices and $\vg(\vp) = (g_1,\dots , g_n)$. We assume without loss of generality that all alternatives have non-negative index and are re-labeled so that $g_1 \ge \dots \ge g_n \ge 0$. By~\Cref{claim:revenue-exp-mandatory}, we have $\optrev =  \sum_{i=1}^n \E[(X_i - g_i)^+]\cdot \prod_{j < i} F_{j}(g_i)$.

Observe that as $g_i$ are non-increasing in $i$, the same is true for the products $\lambda_i:=\prod_{j < i} F_{j}(g_i)$. Furthermore, $\lambda_1 = 1$. We also let $\lambda_{n+1}=0$ for consistency. Then, there exists an index $k\in [n]$ such that $\lambda_k\geq \frac{1}{2}$ and $\lambda_{k+1}< \frac{1}{2}$. We first analyze the revenue contribution of alternatives $\{k+1,\dots , n\}$:

\begin{align*}
    \rev_{k+1}^n &:=\sum_{i=k+1}^n \E[(X_i - g_i)^+] \cdot\prod_{j < i} F_{j}(g_i) = \sum_{i=k+1}^n \E[(X_i - g_i)^+] \cdot \prod_{j < k+1} F_{j}(g_i) \cdot \prod_{j=k+1}^{i-1} F_{j}(g_i) \\
    &\le\sum_{i=k+1}^n \E[(X_i - g_i)^+] \cdot \prod_{j < k+1} F_{j}(g_{k+1}) \cdot \prod_{j=k+1}^{i-1} F_{j}(g_i) = \lambda_{k+1}\cdot \sum_{i=k+1}^n \E[(X_i - g_i)^+] \cdot \prod_{j=k+1}^{i-1} F_{j}(g_i)\\
    &\le \frac{1}{2} \sum_{i=k+1}^n \E[(X_i - g_i)^+] \cdot \prod_{j=k+1}^{i-1} F_{j}(g_i) \le \frac{1}{2} \optrev\\
\end{align*}
where the first inequality follows from the fact that $i\geq k+1$ implies $g_i\le g_{k+1}$ and therefore $F_j(g_i)\le F_j(g_{k+1})$ for all $j$, and the second inequality follows by definition of $\lambda_{k+1}$. Finally, observe that the final sum is precisely equal to the revenue collected by pricing $p_i=\infty$ for all $i\leq k$ (removing them from the buyer's consideration) and pricing all alternatives $i>k$ so that their index is $g_i$ (maintaining the same ordering). Therefore, we conclude that $\rev_{k+1}^n \leq \frac{1}{2}\cdot\optrev$.

We now proceed to analyze the revenue contribution of alternatives $\{1,\dots , k\}$:
\begin{align*}
    \rev_{1}^k &:=\sum_{i=1}^k \E[(X_i - g_i)^+] \cdot 
    \prod_{j < i} F_{j}(g_i) \leq \sum_{i=1}^k \E[(X_i - g_i)^+] \\
    &\le \sum_{i=1}^k \E[(X_i - g_k)^+] \le 2\lambda_k\cdot \sum_{i=1}^k \E[(X_i - g_k)^+] \\
    &= 2\cdot \sum_{i=1}^k \E[(X_i - g_k)^+]\cdot \prod_{j<k}F_j(g_k) \leq 2\cdot \sum_{i=1}^n \E[(X_i - g_k)^+]\cdot \prod_{j<i}F_j(g_k)
\end{align*}
where the first inequality follows from $F_j(x)\leq 1$, the second inequality follows from $g_i\geq g_k$ for all $i\leq k$, the third follows by definition of $\lambda_k$ and the fourth follows by the fact that we only add non-negative terms. Since $\optrev = \rev_{1}^k + \rev_{k+1}^{n}$, we conclude that 
\[\optrev \leq 4\sum_{i=1}^n \E[(X_i - g_k)^+]\cdot \prod_{j<i}F_j(g_k)\]
and the proof follows by observing that the sum in the right-hand side is precisely equal to the revenue achieved by the uniform-index prices $\vp(g_k)$.
\end{proof}

\subsection{Computational Considerations}
\Cref{thm:4-approx-mandatory-unif-id} implies the existence of ``good'' uniform-index prices but does not show how to efficiently obtain them, as the ones used in our analysis are constructed from the optimal (non uniform-index) price vector. The following lemma demonstrates an important structural property of $\pbp$: it is always in the interest of the seller to set the indices of the alternatives at some points in their (combined) supports.

\begin{lemma} \label{lem:gittins-index-in-support-mandatory}
    Fix any $\pbp$ instance $\inst$ and let $S = \{0\} \cup \bigcup_{i=1}^n \support(X_i)$ be the union of the supports of all the alternatives (and zero). Then, for all prices $\vp \in \R^n_+$, there exist prices $\vec{p'}\in\R^n_+$ whose indices are in $S$ (i.e. $g_i(p'_i)\in S$ for all $i\in [n]$) that achieve weakly better revenue: $\rev(\inst, \vec{p'}) \ge \rev(\inst, \vp)$. Similarly, for any (uniform) index $g\in\R$, there exists $g'\in S$ such that $\rev(\inst,\vp(g'))\ge \rev(\inst,\vp(g))$.
\end{lemma}
\begin{proof}
    Let's focus on the first statement. Suppose $\vg = \vg(\vp)$ are the indices corresponding to $\vp$. We re-label the alternatives so that $g_1\ge \dots \ge g_n\ge 0$. Recall that from~\Cref{claim:revenue-exp-mandatory},
    \[\rev(\vp) = \sum_{i =1}^n \E[(X_i - g_i)^+] \prod_{j < i} F_{j}(g_i).\]

    Assume that $\vg$ contains at least one index that does not belong in $S$ and let $g_k$ be the last (minimum) such index. Then, we claim that decreasing $g_k$ to $g'_k$ being the largest value in $S$ that is smaller than $g_k$ weakly increases the revenue. We first note that this transformation still maintains the order of the buyer's strategy: as $g_{k + 1} \in S$ and $g_k \ge g_{k + 1}$, by the choice of $g'_k$ we must necessarily have $g'_k \ge g_{k + 1}$ and $g'_k \le g_{k - 1}$. If $k=n$, then we still have $g'_k\ge 0$ by our definition of $S$.
    
    Consider the change in revenue under this transformation. Observe that the only term that is affected in the sum above is the one corresponding to $i=k$; that is, $\E[(X_k - g_k)^+] \prod_{j < i} F_{j}(g_k)$, which is turned into $\E[(X_k - g'_k)^+] \prod_{j < i} F_{j}(g'_k)$. By the choice of $g'_k$ being the largest value in $S$ that is smaller than $g_k$, we have $F_{j}(g'_k) = F_{j}(g_k)$ for all $j$, and $\E[(X_k - g'_k)^+] \ge \E[(X_k - g_k)^+]$ as $g'_k < g_k$. Therefore, the revenue weakly increases by modifying $g_k$ to $g'_k$. 

    For the second statement regarding the uniform index vector, we apply the same argument starting from an index $g\notin S$ and decreasing it to $g'=\max\{g'\in S: g'<g\}$.
\end{proof}

Furthermore, once the set of indices is specified, computing the optimal tie-breaking rule for the buyer (i.e. the one that maximizes the seller's revenue) is simple.

\begin{lemma} \label{lem:gitins-index-tiebreak-mandatory}
    Fix any $\pbp$ instance $\inst$ and any prices $\vp$ corresponding to indices $\vg$. The tie-breaking order that maximizes the seller's revenue is to inspect all alternatives $i$ of the same index $g_i = g$ in decreasing order of $\expectt{X_i\mid X_i > g}$.
\end{lemma}
The proof of~\Cref{lem:gitins-index-tiebreak-mandatory} follows from a simple exchange argument and is deferred to the appendix. Importantly, \Cref{lem:gitins-index-tiebreak-mandatory} implies that given any set of indices, we can efficiently compute the seller's revenue under the corresponding prices. Combined with~\Cref{lem:gittins-index-in-support-mandatory}, we obtain a simple algorithm for optimizing revenue over the family of uniform-index prices: we simply iterate over all indices in the supports of the alternatives 
and compare the corresponding revenues. For discrete random variables $X_i$, this algorithm runs in polynomial time to the size of the input. Therefore, from~\Cref{thm:4-approx-mandatory-unif-id}, we obtain the following corollary:
\begin{corollary}
\label{cor:efficient}
    For any $\pbp$ instance $\inst$ over discrete alternatives, we can efficiently compute uniform-index prices $\vp(g)$ such that $\optrev(\inst)\leq 4\cdot \rev(\inst,\vp(g))$.
\end{corollary}

For continuous alternatives, if they all have monotone hazard rate then~\Cref{lem:exact-opt-mandatory-special-cases} already implies that the optimal revenue is attained at $g_i=0$ (equivalently, $p_i=\mu_i$). Finally, for arbitrary continuous distributions, computational efficiency is not well-defined as it is not even clear how the distributions are described to the algorithm. In any case, we show in~\Cref{app:discetization} that a simple discretization of their supports (that reduces them to the discrete case) would only incur an additive loss. 
\begin{lemma} \label{lem:discretize-support-mandatory}
    Fix any $\pbp$ instance $\inst$ such that $X_i\le U$ for all $i$ with probability $1$. For some $\eps > 0$, consider the grid $S_\eps = \left\{0, \eps, 2\eps, \dots, \lfloor U/\eps \rfloor\eps, U\right\}$. Let $\vg$ be any index vector, and $\vg'$ be the vector via rounding up every index in $\vg$ to the nearest grid point in $S_\eps$, i.e. $g'_i = \min\{U, \lceil g_i/\eps \rceil \cdot \eps\}$. Then,
    \[\rev(\inst, \vp(\vg')) \ge \rev(\inst, \vp(\vg)) - n\eps.\]

\end{lemma}

\subsection{Uniform Prices}
Finally, in light of the strong performance of uniform-index pricing, it is natural to ask whether uniform (i.e., anonymous) prices of the form $p_i=p$ might enjoy similar guarantees. Unfortunately, this is not the case. As a first obstruction, the optimal uniform price need not coincide with any value in the support of the alternatives. For example, consider $n=3$ identical alternatives taking value $X_i\in \{1,9\}$ with probabilities $0.9$ and $0.1$ respectively. From~\Cref{lem:exact-opt-mandatory-special-cases} we know that the optimal indices (and therefore also the optimal prices) will be uniform. A simple computation implies that indeed, the maximum revenue is uniquely attained by $p_1=p_2=p_3=0.8$ which does not lie in the support of the alternatives. Therefore, even if uniform prices were to achieve optimal revenue, it would not be clear how to optimize over the price $p$.

More importantly, it turns out that uniform prices cannot achieve \textit{any} constant approximation with respect to the optimal revenue, even if the distributions of the alternatives are Bernoulli.

\begin{theorem}\label{thm:uniform-cost-gap}
    For any number of alternatives $n\geq 1$, there exists a $\pbp$ instance $\inst$ over $n$ alternatives with $\optrev(\inst)=\Omega(\log n)\cdot \max_{p\geq 0}\rev(\inst, (p,\dots, p))$.
\end{theorem}
\noindent We prove~\Cref{thm:uniform-cost-gap} in~\Cref{app:unif-prices-gap}.

%% file: mandatory-extensions.tex
\section{Cost-Bearing Sellers and Combinatorial Buyers} \label{sec:mandatory-ext}

In this section, we study several generalizations of $\pbp$. We begin with a modest extension that incorporates (i) an outside option $y\ge 0$, under which the buyer obtains utility $y$ whenever they select no alternative, and (ii) sellers that need to exert some effort $c_i\ge 0$ in order to reveal the value of the $i$-th alternative to the buyer. Our baseline model from~\Cref{sec:preliminaries} is recovered by setting $y=0$ and $c_i=0$ for all $i\in[n]$. Leveraging the robustness of our proof of \Cref{thm:4-approx-mandatory-unif-id}, we show that uniform-index prices remain $4$-approximately optimal and efficiently optimizable in this extension.

We then turn to the more interesting regime of \emph{combinatorial buyers} who may select more than one alternative. We model feasibility via a matroid constraint over the set of alternatives. We show that, for several natural matroid families, one can still efficiently obtain a $4$-approximation to the optimal revenue, even though uniform-index prices are no longer constant-competitive in general. Finally, we note that all our positive results would still hold under outside options and seller costs in the matroid setting; the proofs are straightforward and omitted for brevity.

\subsection{Seller Costs and Outside Options}

An instance $\inst$ is now specified by the $n$ alternatives $X_i$ and their distributions, a cost vector $\vc\in\R_+^n$, and an outside option $y\ge 0$. This information is common knowledge to the buyer and the seller. As before, the seller posts a price vector $\vp\in\R_+^n$, and the buyer sequentially inspects alternatives. To inspect alternative $i$, the buyer pays $p_i$ to the seller, who incurs cost $c_i$ to reveal its realization (equivalently, the seller's net revenue from inspecting $i$ is $p_i-c_i$). At any point, the buyer may stop and either (i) select one of the inspected alternatives $i$, obtaining reward $x_i$, or (ii) take the outside option and obtain reward $y$.

Both $y$ and $\vc$ affect the seller's optimal prices. If $y$ is large, the seller cannot price too aggressively: overpriced alternatives will never be inspected. Formally, treating the outside option as an additional alternative with deterministic value $y$ and price $0$ (and hence index $g=y$), it follows that the buyer will never inspect any alternative $i$ with $g_i<y$. Conversely, the seller should never set $p_i<c_i$, since inspecting $i$ would yield negative net revenue. Thus, if $c_i$ is large, the seller is forced to set a high price for $i$, which may discourage inspection and foreclose otherwise attainable revenue.

Concretely, given a set of prices $\vp\in\R^n_+$, let $\vg$ be the corresponding indices and re-label the alternatives so that $g_1\ge \dots \ge g_n$; as usual, we break ties in favor of the seller. Then, the buyer inspects the alternatives in that order until the maximum of the best realization and the outside option dominates the next index, at which point they halt. Therefore, the corresponding expression for the revenue of the seller is given by the following generalization of~\Cref{claim:revenue-exp-mandatory}:
\[\rev(\inst,\vp) := \sum_{i: g_i\ge y}\left(\E[(X_{i} - g_{i})^+] - c_{i}\right) \cdot\prod_{j < i} F_{j}(g_{i}).\]

As already noted, setting $p_i<c_i$ is never beneficial for the seller and can only decrease revenue, since inspecting such an alternative yields negative net proceeds. Consequently, a \emph{pure} uniform-index pricing rule may fail to approximate the optimal revenue in this extension: a common index may induce inspection of alternatives that generate a net loss for the seller. Nevertheless, we show that a simple post-processing step suffices. Given a uniform-index price vector, we drop all loss-making alternatives by setting their prices to $\infty$ (equivalently, making them unavailable). The resulting pricing rule still achieves a $4$-approximation to the optimal revenue. The proof follows the same blueprint as \Cref{thm:4-approx-mandatory-unif-id} and is included in \Cref{app:seller_cost_proof} for completeness.

\begin{theorem}
    \label{lem:4-approx-mandatory-unif-id-sellercost}
     For any $\pbp$ instance $\inst$ with seller costs $\vc$ and outside option $y$, we can efficiently compute prices $\vp\in\R^n_+$ with $\vg(\vp)\in\{g,-\infty\}^n$ such that $\optrev(\inst)\leq 4\cdot \rev(\inst, \vp)$.
\end{theorem}

\subsection{Matroid Selection}

We now consider a combinatorial generalization of our model in which the buyer may select \emph{multiple alternatives}, obtaining their additive value as a reward. In particular, the $n$ alternatives correspond to the elements of an underlying \textit{matroid} $\matroid = ([n],\independent).$\footnote{That is, $\independent$ is a collection of subsets of $[n]$ such that (i) $\emptyset\in\independent$; (ii) for all $S,S'\subseteq [n]$ with $S\subseteq S'$ we have $S'\in\independent\Rightarrow S\in\independent$; and (iii) for all $S,S'\in \independent$ with $|S|<|S'|$, there exists $i\in S'\setminus S$ such that $S\cup\{i\}\in\independent$.} The seller posts prices $\vp\in\R_+^n$, and the buyer sequentially inspects the alternatives. At any point, the buyer may stop and choose any subset $S$ of the inspected alternatives that is independent ($S\in\independent$), obtaining reward $\sum_{i\in S} X_i$. The single-selection setting studied so far is recovered as the special case of a $1$-uniform matroid, i.e.,
$\independent=\{S\subseteq[n]: |S|\le 1\}$.

\citet{S18} introduced the corresponding buyer problem and showed that Weitzman's indices continue to characterize the buyer's optimal behavior. Specifically, for each alternative $i$, define its index $g_i$ as before and relabel the alternatives so that $g_1\ge \dots \ge g_n$ (breaking ties in favor of the seller). For each $i$, let
\[
S_i \;:=\; \{\, j<i : \text{$j$ was inspected and } X_j > g_i \,\}
\]
denote the (random) set of earlier inspected alternatives whose realized value exceeds the threshold $g_i$. The utility-maximizing buyer processes alternatives in this order and inspects alternative $i$ \emph{only if} $i$ is \emph{not spanned} by $S_i$ in $\matroid$. After iterating through all elements, the buyer selects a maximum-weight independent set among the inspected alternatives (with weights given by their realizations). As before, the buyer may opt out and receive a utility of $0$; consequently, alternatives with $g_i<0$ are never inspected. Therefore, it follows that the seller's expected revenue under an index vector $\vg$ in the matroid-selection setting can be written as
\[
\rev(\inst,\vp(\vg))
:= \sum_{i:\, g_i\ge 0}
\expectt{(X_i-g_i)^+}\cdot
\probb{i \notin \spn(S_i)}.
\]

The inspection probabilities in the above expression behave quite differently than in the single-selection case (\Cref{claim:revenue-exp-mandatory}), where the probability of inspecting alternative $i$ depends only on $g_i$ and the distributions of earlier alternatives $j<i$. Under a general matroid, however, it also depends on the \emph{identity} of element $i$---more precisely, on how $i$ relates to earlier elements through the matroid structure. As a result, the inspection probabilities need not vary monotonically with $i$. This loss of monotonicity breaks a key ingredient in the proof of the $4$-approximate optimality of uniform-index prices in the single-selection setting (\Cref{thm:4-approx-mandatory-unif-id}). 

On the positive side, we show that for some special families of matroids, one can still compute prices that can approximate the optimal revenue up to a constant factor. We begin with $k$-\textit{uniform matroids}, which model the setting where the buyer may ultimately select any subset $S$ of at most $k$ inspected alternatives. In this case, the probability that an alternative is inspected remains monotone with respect to the index order and to the index itself. Exploiting this monotonicity, and adapting the arguments from \Cref{thm:4-approx-mandatory-unif-id}, we prove the following (see~\Cref{app:uniform-matroid-proof} for the details):

\begin{theorem}
    \label{lem:4-approx-mandatory-unif-id-unifmatroid}
     For any $\pbp$ instance $\inst$ over a uniform matroid constraint $\matroid$, we can efficiently compute a uniform-index pricing $\vp(g)$ such that $\optrev(\inst)\leq 4\cdot \rev(\inst, \vp(g))$.
\end{theorem}

We then shift our attention to \emph{partition matroids}. The ground set of alternatives $[n]$ is partitioned into $m$ components $C_1,\dots,C_m$, and the buyer may select any set $S$ satisfying $|S\cap C_j|\le k_j$ for all components $j\in [m]$. Observe that, for any fixed price vector $\vp$, the buyer's utility maximization problem \emph{decomposes} across components: the buyer can optimize independently within each $C_j$, treating it as a $k_j$-uniform matroid. Since \Cref{lem:4-approx-mandatory-unif-id-unifmatroid} enables us to (approximately) optimize revenue extraction from each component, the following result follows immediately:

\begin{corollary}
    \label{corr:4-approx-mandatory-unif-id-partitionmatroid}
    For any $\pbp$ instance $\inst$ over a partition matroid constraint $\matroid$, we can efficiently compute a price vector $\vp$ such that $\optrev(\inst)\leq 4\cdot \rev(\inst,\vp)$.
\end{corollary}

Note that the indices induced by the above approach need not be \emph{globally} uniform: while each component of the partition matroid admits an (approximately) optimal uniform-index price, the resulting indices may differ across components. This naturally raises the question of whether a \emph{single} uniform-index pricing can still guarantee a constant-factor approximation to the optimal revenue for general matroids. We answer this question in the negative even for partition matroids. In particular, we show that the approximation gap between uniform-index and optimal pricings can grow linearly with the number of components.

\begin{theorem} \label{lem:k-gap-partition-matroid-unif-id}
    For any $m\ge 1$, there exists a $\pbp$ instance $\inst$ over a partition matroid constraint with $m$ components, $n=\Theta(m \cdot 2^{m^2})$ elements and $k_j=1$ for all $j\in [m]$, such that
    \[\optrev(\inst) \ge \Omega(m) \cdot \max_{g \in \R}\rev(\inst, \vp(g)).\]
\end{theorem}

The following corollary is also immediate from our construction.

\begin{corollary}\label{corr:sqrt-logn-gap-partition-matroid-unif-id}
    For any $n\ge 1$, there exists a $\pbp$ instance $\inst$ over $n$ alternatives and a partition matroid constraint such that
    \[\optrev(\inst) \ge \Omega(\sqrt{\log n}) \cdot  \max_{g \in \R}\rev(\inst, \vp(g)).\]
\end{corollary}

As already mentioned, going beyond partition matroids presents significant challenges. Not only are uniform-index pricings no longer a good option, but the underlying matroid enforces strong correlations between the inspection probabilities of the different alternatives. We leave the design of (approximately) optimal prices for general matroids as an intriguing open question. The remainder of this section is dedicated towards proving our lower bound for partition matroids.

\begin{proof}[Proof of~\Cref{lem:k-gap-partition-matroid-unif-id}]

As the theorem states, our matroid will be composed of $m$ components $C_j$ and the buyer will only be able to select $k_j=1$ alternative from each of them. Furthermore, each component will be composed of $n_j$ identically distributed alternatives, following the same distribution $\dist_j$. Let $\inst$ be used to denote the (entire) matroid selection instance and $\inst_j$ be used to denote the single-selection instance corresponding to the $j$-th component (i.e. with $n_j$ alternatives of distribution $\dist_j$). Since the optimal buyer will optimize over each component independently, and each component contains identically distributed alternatives, we know that the optimal prices will be uniform-index on each component (\Cref{lem:exact-opt-mandatory-special-cases}). Furthermore, the seller's total revenue decomposes to the sum of revenues achieved in each component. Therefore, we have

\begin{equation}\label{eq:partition-decomposition}
    \frac{\optrev(\inst)}{\max_{g\ge 0}\rev(\inst,\vp(g))} = \frac{\sum_{j=1}^m\optrev(\inst_j)}{\max_{g\ge 0}\rev(\inst,\vp(g))} = \frac{\sum_{j=1}^m  \max_{g\ge 0}\rev(\inst_j,\vp(g))}{\max_{g\ge 0}\sum_{j=1}^m \rev(\inst_j,\vp(g))} .
\end{equation}

To complete our proof, we will define the instances $\inst_j$ in a manner so that we can extract non-trivial revenue from each component $C_j$ only for indices $g$ in some interval $I_j$; by making these intervals disjoint across the different components, the proof will follow.

We now define the instances $\inst_j=(n_j,\dist_j)$. For each $j\in [m]$, let $\ell_j = m\cdot j$ and $r_j = m\cdot (j + 1)$. We set $n_j = m \cdot 2^{\ell_j}$ and define distribution $\dist_j$ such that
\[
X\sim\dist_j = 
\begin{cases}
    x& \text{w.p. $\; 2^{-x-1}$ for all $x \in \{0,1,\dots, \ell_j-1\}$} \\
    \ell_j & \text{w.p. $\; (1-1/m)\cdot 2^{-\ell_j}$} \\
    r_j& \text{w.p. $\; n_j^{-1}$}
\end{cases}
\]
Note that $\sum_{x=0}^{\ell_j-1}2^{-x-1} + (1-1/m)\cdot 2^{-\ell_j} + n_j^{-1}=1$ by definition of these quantities. Also note that the total number of alternatives across all components is 
$n=\sum_{j=1}^m n_j = \sum_{j=1}^m m\cdot 2^{mj} = \Theta(m\cdot 2^{m^2})$.

We will next compute the revenue $\rev(\inst_j, \vp(g))$ that a uniform-index pricing achieves on instance $\inst_j$. A straightforward application of~\Cref{claim:revenue-exp-mandatory} for identical alternatives implies
\begin{align*}
    \rev(\inst_j, \vp(g)) &= 
     \expectt{X-g \mid X > g}\cdot (1 - \probb{X\le g}^{n_j})
\end{align*}
for $X\sim\dist_j$. Since all values in the support of $\dist_j$ are integers, we can focus on integer indices without loss of generality (see~\Cref{lem:gittins-index-in-support-mandatory}). Observe that:
\begin{itemize}
    \item For $g\ge r_j$, we have $\rev(\inst_j, \vp(g)) = 0$.

    \item For $g\in \{\ell_j, \dots, r_j-1\}$, we have $\rev(\inst_j, \vp(g)) = (r_j-g)\cdot (1-(1-1/n_j)^{n_j}) = \Theta(1)\cdot (r_j-g)$.
    
    \item For $g\in\{0,1,\dots , \ell_j-1\}$ we have that $\probb{X\le g}=1-2^{-g-1}$ and also
    \begin{align*}
    \E[X - g \mid X > g] &= 2^{g+1} \left(\sum_{x = g+1}^{\ell_j - 1} \frac{x - g}{2^{x + 1}} + (\ell_j - g) \frac{m - 1}{m \cdot 2^{\ell_j}} + (r_j - g) \frac{1}{m \cdot 2^{\ell_j}} \right) \\
    &= 2^{g+1} \left(\sum_{x = 1}^{\ell_j - g-1} \frac{x}{2^{x + g + 1}} +  (\ell_j - g)\frac{1}{2^{\ell_j}} +  (r_j - \ell_j)\frac{1}{m2^{\ell_j}}\right) \\
    &= \sum_{x = 1}^{\ell_j - g-1} \frac{x}{2^{x}} + \frac{1+\ell_j - g}{2^{\ell_j - g - 1}} = 2
\end{align*}
with the last equality following from the identity $\sum_{j=1}^{x-1} \frac{j}{2^j} + \frac{x+1}{2^{x-1}} = 2$ for $x \in \mathbb{Z}_{\ge 1}$. Therefore, we have that $\rev(\inst_j,  \vp(g))=O(1)$.
\end{itemize}

Since the instance $\inst_j$ is i.i.d., by~\Cref{lem:exact-opt-mandatory-special-cases}, we obtain that the optimal revenue for $\inst_j$ is given by $\optrev(\inst_j) = \max_{g\ge 0}\rev(\inst_j,\vp(g)) = \Theta(m)$, achieved by $g=\ell_j$. Furthermore, the revenue achieved by any index $g\notin [\ell_j, r_j)$ is $O(1)$. Since the intervals $[\ell_j, r_j)= [j\cdot m,(j+1)\cdot m)$ do not overlap between different components $C_j$, this implies that any (globally) uniform-index can only extract $\Theta(m)$ revenue from one component and $O(1)$ revenue from the remaining ones. Together with~\Cref{eq:partition-decomposition}, these imply the theorem.
\end{proof}

%% file: optional-inspection.tex
\section{Beyond Mandatory Inspection: The Optional Inspection Setting} \label{sec:optional}

We now shift our attention to the more challenging setting of \emph{optional inspection}. As in the mandatory inspection model, there are $n$ stochastic alternatives with distributions $\dist_i$ that are known to both the buyer and the seller; the seller announces a price vector $\vp \in \R^n_+$; the buyer may pay $p_i$ to observe the realization of alternative $i$; and the buyer ultimately selects a single alternative so as to maximize their utility. The only difference is that, under optional inspection, the buyer may also select alternative $i$ \emph{without inspecting it}, in which case they obtain the (random) reward $X_i$---equivalently, its expected value $\mu_i$---without making any payment to the seller.

Although this modification may appear innocuous at first glance, it is in fact consequential: \citet{FLL22} show that even for three-point distributions, the buyer's utility-maximization problem is NP-hard. This stands in sharp contrast to the mandatory inspection setting, where the optimal buyer response can be computed efficiently for arbitrary distributions. Consequently, the seller's revenue-maximization problem becomes substantially more challenging: we no longer have a clean characterization of the buyer's best response, nor a closed-form expression---analogous to \Cref{claim:revenue-exp-mandatory}---for the revenue induced by a given pricing $\vp$.

Indeed, the optimal revenue achievable under mandatory and optional inspection can differ sharply, even for the same set of alternatives. To illustrate, when $n=1$, the seller can extract revenue $\mu_1$ in the mandatory setting by setting $p_1=\mu_1$; in contrast, in the optional setting the buyer will never pay any strictly positive price, since $X_1 \ge 0$ implies that they can always obtain nonnegative utility by selecting the alternative without inspection. 
Moreover, several useful and intuitive properties from the mandatory setting may fail under optional inspection: for instance, increasing the price of an alternative $j$ can \emph{decrease} the probability that the buyer inspects a different alternative $i \neq j$, even when all other prices are held fixed. We provide such an example in~\Cref{app:proofs-for-optional-examples}.

In this section, we initiate the study of revenue maximization under optional inspection and develop approximately optimal mechanisms for two special cases: identically distributed alternatives and two-point distributions. We remark that in the intersection of these two cases (i.e. identically distributed two-point distributions) uniform prices always achieve the optimal revenue; a closed-form solution is derived in~\Cref{app:proofs-for-optional-examples}.

\subsection{Identically Distributed Alternatives}\label{sec:optional-iid}

We first consider the special case in which all alternatives are identically distributed, i.e., $X_i\sim\dist$ for every $i\in[n]$. As a warm-up, we first consider $n=2$ alternatives (recall that when $n=1$, the seller cannot extract any positive revenue under optional inspection). Perhaps unsurprisingly, the optimal pricing in this case is uniform, and the optimal price can be computed efficiently.

\begin{theorem}\label{thm:optional-iid-2}
    For any optional inspection instance $\inst$ over $n=2$ identical alternatives, we can efficiently compute a uniform price vector $\vp = (p,p)$ that achieves optimal revenue.
\end{theorem}

For $n \ge 3$, it is no longer clear whether uniform prices remain optimal in the optional inspection model (though we conjecture that they do). Regardless, the symmetry of the alternatives can still be leveraged to design ``nearly uniform'' prices that are efficiently computable and achieve ``nearly optimal'' revenue as the number of alternatives $n$ grows.

\begin{theorem} \label{thm:optional-iid-approx}
    For any optional inspection instance $\inst$ over $n$ identical alternatives, we can efficiently compute a price vector $\vp=(\infty,p,p,\dots ,p)$ such that
    \[\rev(\inst,\vp) \geq \left(1-\frac{1}{n}\right)\cdot \optrev(\inst).\]
\end{theorem}

The proof proceeds by setting the price of one alternative to $\infty$, effectively preventing the buyer from ever inspecting it. Since the alternatives are identical, selecting any alternative without inspection yields the same expected reward $\mu$ to the buyer; thus, without loss of generality, we may assume that the only alternative the buyer can select without inspection is precisely the one priced at $\infty$. Therefore, the task of pricing the remaining alternatives $2,\dots,n$ reduces to a mandatory inspection instance with an outside option of $\mu$. As established in previous sections, such settings admit efficiently computable uniform prices that are optimal. It remains to argue that sacrificing one alternative can reduce the optimal revenue by at most a $\tfrac{1}{n}$ fraction. While this claim is intuitive, proving it directly is subtle, since even with identical alternatives the optional inspection model lacks a simple characterization of the buyer's optimal policy. We therefore avoid a direct analysis and instead apply a reduction from the optional to the mandatory inspection setting, where the desired bound is straightforward to establish.

The detailed proofs of~\Cref{thm:optional-iid-2,thm:optional-iid-approx} are presented in~\Cref{app:proofs-for-optional-iid}.

\subsection{Two-Point Distributions and Reveal Mechanisms}
We now shift our attention to the special case where each alternative $X_i$ takes two distinct values $\{l_i,r_i\}$ such that $0\leq l_i \le r_i$ with probabilities $(1-q_i)$ and $q_i$ respectively, for some $q_i\in [0,1]$.

To approximate the optimal revenue in such optional inspection settings, we broaden our focus to a more general class of (partially) adaptive mechanisms, which we call \emph{reveal mechanisms}. As the name suggests, a reveal mechanism may disclose the realizations of an arbitrary subset of alternatives to the buyer for free. Doing so forfeits any direct revenue from these alternatives, since once revealed they can be selected without further payment. The key feature, however, is that the seller is then permitted to price access to the remaining alternatives \emph{as a function of the revealed values}, potentially extracting more revenue from them compared to non-adaptive prices. Moreover, among the revealed alternatives the buyer will only ever consider choosing the one with the highest realized value. Consequently, without loss of generality, the seller's contingent pricing rule needs only to depend on the revealed information through this maximum value. Formally:

\begin{definition}[Reveal Mechanisms]
    A reveal mechanism $\mech$ is characterized by a set $R\subseteq [n]$ of alternatives and a mapping $p(\cdot)$ from real numbers $y\geq 0$ to prices over $[n]\setminus R$. Initially, all realizations $x_i$ for $i\in R$ are revealed to both the seller and the buyer for free. The seller then proceeds to set prices $p_i(y_R)$ on the alternatives $i\notin R$ where $y_R:=\max_{i\in R}x_i$. Finally, the buyer proceeds to maximize their utility by potentially inspecting alternatives $i\notin R$, having the option to select any alternative $i\in [n]$ without inspection at any point. We use $\revealset$ to denote the set of all reveal mechanisms and $\rev(\inst,\mech)$ for the expected revenue of mechanism $\mech\in\revealset$ on instance $\inst$.
\end{definition}

We note that the definition of reveal mechanisms extends beyond the special case of two-point distributions and is natural from an economics perspective: by disclosing a subset of realizations at no charge, the seller can infer how valuable additional information is likely to be and then set prices that are better aligned with the buyer’s posterior willingness to pay. We also note that while one could define reveal mechanisms in the mandatory inspection model as well, revealing information for free is \textit{never beneficial} to the seller in that setting. We prove this formally in~\Cref{app:proofs-for-optional-reveal-mandatory}.


Back to the special case of optional inspection under two-point distributions, our main result for this section is the following:
\begin{theorem}\label{thm:optional-bernoulli}
    For any optional inspection instance $\inst$ over two-point distributions, we can efficiently compute a reveal mechanism $\mech$ such that 
    \[\rev(\inst,\mech)\geq \frac{1}{2}\cdot \max_{\mech'\in \revealset}\rev(\inst,\mech').\]
\end{theorem}
\begin{proof}
    Let $\mu_i = l_i\cdot (1-q_i) + r_i\cdot q_i$ denote the expectation of the $i$-th alternative and assume without loss of generality that  $\mu_1 = \max_{i\in [n]}\mu_i$. Observe that under any mechanism $\mech\in \revealset$, the utility of the buyer is at least $\mu_1$ as they can always select the first alternative without inspecting anything, and the (expected) value of the buyer is at most $\expectt{\max_i X_i}$ as in the best case, they always select the alternative of maximum value. Since the buyer is utility-maximizing, this implies that
    \[\max_{\mech'\in \revealset}\rev(\inst,\mech') \leq \expectt{\max_i X_i}-\mu_1.\]
    Our proof entails the (efficient) construction of two reveal mechanisms $\mech_1$ and $\mech_2$ whose combined revenue matches this upper bound - then, by selecting the best out of the two (or uniformly randomizing over our choice), the proof follows. The first mechanism will exploit the option of adapting to revealed realizations and will be responsible for extracting maximum revenue from the first alternative. The second mechanism will be responsible for extracting the revenue from alternatives $2$ through $n$ - since we don't care about the first alternative, we can price it at $p_1=\infty$ to ensure that the buyer never inspects it and that it is the unique alternative that potentially gets selected without inspection (as it has the maximum expectation). This essentially allows us to treat the remaining pricing problem as in the mandatory inspection setting and exploit the optimality of uniform-index prices for two-point distributions. We formalize these statements in the following two claims (proved in~\Cref{app:proofs-for-optional-bernoulli}), from which the proof of~\Cref{thm:optional-bernoulli} follows.
\end{proof}

    \begin{claim}\label{claim:bernoulli_mech_1}
        Let $\mech_1$ be the mechanism that reveals alternatives $R=[n]\setminus\{1\}$ and prices the first alternative at $p_1 = \expectt{\max\{X_1,y_R\}} - \max\{\mu_1,y_R\}$. Then,
        \[\rev(\inst,\mech_1) = \expectt{\max_{i\in [n]}X_i} - \expectt{\max\{\mu_1, \max_{i\ge 2}X_i\}}.\]
    \end{claim}

      \begin{claim}\label{claim:bernoulli_mech_2}
        Let $\mech_2$ be the (non-adaptive) pricing mechanism corresponding to $R=\emptyset$ and prices $\vp=(p_1,\dots , p_n)$  where $p_1 = \infty$ and prices $p_j\ge 0$ satisfy $g_j(p_j)=\mu_1$ for all $j\ge 2$. Then,
        \[\rev(\inst,\mech_2) = \expectt{\max\{\mu_1, \max_{i\ge 2}X_i\}} - \mu_1.\]
    \end{claim}

%% file: conclusion.tex
\section{Conclusion}\label{sec:conclusion}

In this work, we initiated the study of revenue maximization in settings where a seller prices access to information and a buyer responds by solving a Pandora's Box sequential search problem. We identified a natural class of pricing mechanisms---\emph{uniform-index} pricings, which equalize the induced Weitzman index across alternatives---and proved that this class attains a constant-factor approximation to the optimal revenue while remaining amenable to efficient optimization. 



Our work suggests many directions for future investigation, including tightening the approximation factors we achieve. Two directions are particularly exciting, but appear challenging: approximating revenue when a mandatory inspection buyer faces a general matroid constraint; and approximating revenue from an optional inspection buyer with arbitrary alternative distributions.

%% file: appendix.tex
\input{appendix_mandatory}
\input{appendix_mandatory_ext}

\input{appendix_optional}

\input{appendix_examples}

%% file: appendix_mandatory.tex
\section{Omitted Proofs from~\Cref{sec:mandatory}}\label{app:proofs-for-mandatory}

\subsection{Optimal Prices Under Special Instances (Proof of~\Cref{lem:exact-opt-mandatory-special-cases})}\label{appendix:mandatory-special-proofs}

The proof of this theorem comes from the following two lemmas, combined with the cases of two-point and MHR distributions that we analyzed in the main body.

\begin{lemma}\label{lem:2-alternatives}
    If there are only 2 alternatives, then a uniform-index pricing is optimal.
\end{lemma}

\begin{proof}
    Consider the index vector $\vg$ corresponding to the optimal prices, and order the alternatives so that $g_1 \ge g_2$. Note that the revenue contribution of alternative $1$ is exactly $\E[(X_1 - g_1)^+]$. This quantity weakly increases by setting $g_1$ to $g_2$. Furthermore, the probability of inspecting alternative $2$ is $\probb{X_1 \le g_2}$ and is not affected by the choice of $g_1\ge g_2$, proving our claim.
\end{proof}

\begin{lemma}\label{lemma:mandatory-iid}
    If all alternatives are identically distributed, then a uniform-index pricing is optimal.
\end{lemma}
\begin{proof}
    Let $X$ denote the random reward of the $n$ identically distributed alternatives, let $F$ be the common CDF and let $E(z):=\E[(X-z)^+]$. Now fix any optimal pricing $\vp$ and let $\vg$ be the corresponding indices with $g_1\geq \dots \geq g_n \ge 0$. Assume that $\vg$ is not uniform and let $i\in [n]$ be the smallest index for which $g_i>g_{i+1}$, so we have that $g_1=g_2=\dots = g_i > g_{i+1}\geq g_{i+2}\geq\dots \geq g_n$. 

    Now, let $\vp_+$ and $\vp_-$ be the prices corresponding to indices $\vec{g_+}=(g_1,\dots, g_i, g_i, g_{i+2},\dots , g_n)$ and $\vec{g_-}=(g_1,\dots, g_{i-1},g_{i+1}, g_{i+1}, g_{i+2},\dots , g_n)$ respectively. If we can show that either $\vp_+$ or $\vp_-$ achieve at least as large revenue as $\vp$ then we are done; repeating this argument will eventually result in uniform indices. In particular, if $\vp_+$ is better then we have made the first $i+1$ indices identical (instead of the first $i$ indices) and if $\vp_-$ is better then we can repeat the same argument to set $g'_1=g'_2=\dots =g'_{i}=g_{i+1}$ and make the first $i + 1$ indices identical~--- this last part assumes that the argument will not depend on the value of $i$.

    For simplicity, let $a=g_1=g_2=\dots = g_i$ and $b=g_{i+1}$. We have:
    \[\rev(\vp)=E(a) + E(a)F(a) + E(a)F^2(a) + \dots + E(a)F^{i-1}(a) + E(b)F^i(b) + \dots \]

     \[\rev(\vp_+)=E(a) + E(a)F(a) + E(a)F^2(a) + \dots + E(a)F^{i-1}(a) + E(a)F^i(a) + \dots \]

    \[\rev(\vp_-)=E(a) + E(a)F(a) + E(a)F^2(a) + \dots + E(b)F^{i-1}(b) + E(b)F^i(b) + \dots \]
    where the last dotted part corresponds to what happens after the first $i+1$ alternatives have been inspected; since the remaining indices are the same in all cases these terms cancel out. Therefore, we want to show that
    \[E(a)F^{i-1}(a) + E(b)F^i(b) \leq \max\bigg(E(a)F^{i-1}(a) + E(a)F^i(a), E(b)F^{i-1}(b) + E(b)F^i(b)\bigg)\]
    for all distributions $X$ and all $a>b$, independently of $i$. Assume that this is not true. Then, we have that
    \[E(b)F^i(b)> E(a)F^i(a)\quad\quad\text{and}\quad\quad E(a)F^{i-1}(a) > E(b)F^{i-1}(b)\]
    which implies that $F(b)>F(a)$, which is a contradiction as $a>b$.
\end{proof}

\subsection{Uniform-Index Prices Are Not Optimal (Proof of~\Cref{thm:const-gap-mandatory-unif-id})}\label{app:unif-index-lb}

Fix a large integer $n$. We consider an instance $\inst$ over $2n$ alternatives $X_1,\dots , X_n$ and $Y_1,\dots , Y_n$ partitioned into two bundles of identically distributed alternatives with probability mass functions
\[
X_i = 
\begin{cases}
1 &\text{w.p. }\; \alpha/n \\
z &\text{w.p. }\; \beta/n \\
0 &\text{otherwise} \\
\end{cases}
\quad \text{ and } \quad
Y_i = 
\begin{cases}
1 &\text{w.p. }\; \gamma/n \\
z &\text{w.p. }\; \delta/n \\
0 &\text{otherwise} \\
\end{cases}
\]
for some $z\in (0,1)$ and parameters $\alpha,\beta,\gamma,\delta>0$ with $\alpha+\beta < n$ and $\gamma+\delta < n$ to be specified later. 

As we establish in~\Cref{lem:gittins-index-in-support-mandatory}, the optimal uniform-index pricing will set an index $g\in\{0,z,1\}$; setting $g=1$ corresponds to prices $p_i=0$ for all $i$, therefore we can exclude that option as it achieves a revenue of $0$. Furthermore, from~\Cref{lem:gitins-index-tiebreak-mandatory}, we can assume without loss that under a uniform-index pricing, alternatives of the same type will be inspected in sequence, i.e. the buyer inspects all alternatives $X_i$ before any alternative $Y_j$ or vice versa.

When we set $g = z$, either order gives the optimal revenue here. Assuming the $X$ alternatives go first, then the expected revenue under $g = z$ is

\begin{align*}
    \rev_z &= \sum_{i=1}^n \expectt{(X_i-z)^+}\cdot \prod_{j=1}^{i-1}\probb{X_i\leq z} + \left(\prod_{i=1}^n \probb{X_i\leq z}\right)\cdot \sum_{i=1}^n \expectt{(Y_i-z)^+}\cdot \prod_{j=1}^{i-1}\probb{Y_i \leq z} \\
    &= (1-z)\cdot\left[1-\left(1 - \frac{\alpha}{n}\right)^n\cdot \left(1 - \frac{\gamma}{n}\right)^n\right].
\end{align*}

\noindent On the other hand, for $g = 0$, the revenue of sending either $X$ or $Y$ alternatives first are

\begin{align*}
    \rev_{0,X} &= \sum_{i=1}^n \expectt{(X_i-0)^+}\cdot \prod_{j=1}^{i-1}\probb{X_i= 0} + \left(\prod_{i=1}^n \probb{X_i= 0}\right)\cdot \sum_{i=1}^n \expectt{(Y_i-0)^+}\cdot \prod_{j=1}^{i-1}\probb{Y_i = 0} \\
    &=  \frac{\alpha + \beta z}{\alpha+\beta}\cdot\left[1-\left(1 - \frac{\alpha+\beta}{n}\right)^n\right] + \left(1 - \frac{\alpha+\beta}{n}\right)^n\cdot\frac{\gamma + \delta z}{\gamma+\delta}\cdot\left[1-\left(1 - \frac{\gamma+\delta}{n}\right)^n\right]
\end{align*}
and
\begin{align*}
    \rev_{0,Y} &= \frac{\gamma +\delta z}{\gamma+\delta}\cdot\left[1-\left(1 - \frac{\gamma+\delta}{n}\right)^n\right] + \left(1 - \frac{\gamma+\delta}{n}\right)^n\cdot\frac{\alpha + \beta z}{\alpha+\beta}\cdot\left[1-\left(1 - \frac{\alpha+\beta}{n}\right)^n\right]
\end{align*}
respectively. Which of the two orders is preferred by the seller depends on whether $\E[X_i \mid X_i > 0] > \E[Y_i \mid Y_i > 0]$ or not. Therefore, we conclude that $\max_{g\in\R}\rev(\vp(g)) = \max\{\rev_z, \rev_{0,X}, \rev_{0,Y}\}$.

Up next, we lower bound the optimal revenue by the best of two non uniform-index mechanisms: one that sets $g_i=z$ for all $X$-alternatives and $g_i=0$ for all $Y$-alternatives, and one that sets $g_i=0$ for all $X$-alternatives and $g_i=z$ for all $Y$-alternatives. For the former pricing, we compute the corresponding revenue to be
\begin{align*}
    \rev_1 &= \sum_{i=1}^n \expectt{(X_i-z)^+}\cdot \prod_{j=1}^{i-1}\probb{X_i\leq z} + \left(\prod_{i=1}^n \probb{X_i=0}\right)\cdot \sum_{i=1}^n \expectt{(Y_i-0)^+}\cdot \prod_{j=1}^{i-1}\probb{Y_i = 0} \\
    &= (1-z)\cdot\left[1-\left(1 - \frac{\alpha}{n}\right)^n\right] + \left(1 - \frac{\alpha+\beta}{n}\right)^n\cdot\frac{\gamma + \delta z}{\gamma+\delta}\cdot\left[1-\left(1 - \frac{\gamma+\delta}{n}\right)^n\right]
\end{align*}
and similarly, the revenue obtained from the second pricing is
\begin{align*}
    \rev_2 
    &= (1-z)\cdot\left[1-\left(1 - \frac{\gamma}{n}\right)^n\right] + \left(1 - \frac{\gamma+\delta}{n}\right)^n\cdot\frac{\alpha + \beta z}{\alpha+\beta}\cdot\left[1-\left(1 - \frac{\alpha+\beta}{n}\right)^n\right]\quad\quad\quad\quad\quad\quad\;\;\;
\end{align*}
from which we obtain $\optrev\geq \max(\rev_1,\rev_2)$. 

We therefore conclude that 
\[\frac{\optrev}{\max_{g\in\R}\rev(\vp(g))}\geq \frac{\max\{\rev_1,\rev_2\}}{\max\{\rev_z, \rev_{0,X}, \rev_{0,Y}\}}=:L(n,z,\alpha,\beta,\gamma,\delta)\]
as from the previous equations, for all parameters $n\geq 1$, $z\in (0,1)$ $\alpha,\beta,\gamma,\delta>0$ with $\alpha+\beta < n$ and $\gamma+\delta < n$. Unfortunately, we don't know how to optimize this lower bound analytically; even if we take $n\rightarrow \infty$ so that $(1-\lambda/n)^n\approx e^{-\lambda}$, these expressions are extremely non-linear and no closed form optimizer could be found. 

Instead, we implemented the lower bound expression in Python and performed a greedy grid-search to find a good set of parameters. Even for reasonable parameters, we obtain that uniform-index mechanisms are suboptimal in general: setting $z=0.01$, $\alpha = 0.033$, $\beta=0.8$, $\gamma = 0.0001$, $\delta=6$ and $n=1000$ yields $L(n,z,\alpha,\beta,\gamma,\delta)\geq 1.131$. By expanding our optimization grid to include more extreme parameters, we conclude our proof by obtaining $L(n,z,\alpha,\beta,\gamma,\delta)>1.14$ for $z=6.82\cdot 10^{-7}$, $\alpha = 7.5\cdot 10^{-11}$, $\beta = 1.2\cdot 10^{4}$, $\gamma = 2.27\cdot 10^{-6}$, $\delta = 7.625\cdot 10^{-1}$ and $n=10^6$.

\subsection{Optimal Ordering of the Alternatives (Proof of~\Cref{lem:gitins-index-tiebreak-mandatory})}

Fix any instance $\inst$, and consider any index vector $\vg\in\R^n_+$; we re-label the alternatives so that the buyer inspects them in increasing order of label; note, $g_1 \ge g_2 \ge \dots \ge g_n$. Under this buyer's order, the revenue obtained by the seller is
\[A = \sum_{i=1}^n \E[(X_i - g_i)^+] \prod_{j < i} F_{j}(g_i).\]

Now assume that there exists $k$ such that $g_k = g_{k + 1} = g$. By swapping the buyer's order for alternatives $k$ and $k+1$, their utility will not be affected whereas the seller's new revenue becomes

\[B = A + \E[(X_k - g)^+] \prod_{j < k} F_{j}(g) (F_{{k+1}}(g) - 1) + \E[(X_{k + 1} - g)^+] \prod_{j < k} F_{j}(g) (1 - F_{k}(g)),\]
since the only changes in the revenue expression are the probabilities of inspecting alternatives $k$ (which is now possibly blocked by alternative $k + 1$) and $k+1$ (which is now no longer blocked by alternative $k$). In order for this swap to weakly improve the revenue (i.e. $B\ge A$), we must have

\[ \E[(X_{k + 1} - g)^+] \cdot (1 - F_{k}(g)) \ge \E[(X_k - g)^+] \cdot (1-F_{{k+1}}(g))\]
which translates to
\[
\frac{\E[(X_{k + 1} - g)^+]}{1-F_{{k+1}}(g)}\ge \frac{\E[(X_k - g)^+]}{1 - F_{k}(g)}
\]
or equivalently, $\E[X_{k + 1} \mid X_{k + 1} > g] \ge \E[X_k \mid X_k > g]$. We therefore arrive at our claim, since this comparison is transitive.

\subsection{Discretization of Weitzman Indices (Proof of~\Cref{lem:discretize-support-mandatory})}\label{app:discetization}

Fix any price vector $\vp$ and its corresponding index vector $\vg$. Assume the alternatives are ordered such that the buyer's strategy goes from alternative $1$ to alternative $n$. Note that the revenue contribution of alternative $i$ is exactly
\[\E[(X_i - g_i)^+] \prod_{j < i} F_{j}(g_i).\]

Now, construct $\vg'$ by rounding up $\vg$ to the grid points in $S_{\eps}$. Specifically, we set the index of each alternative $i$ to $g'_i = \min\{U, \lceil g_i/\eps \rceil \eps\}$. Observe that we still have $g'_1\ge \dots \ge g'_n$, so the order from $1$ to $n$ is still compatible under $\vg'$. Furthermore, as $g_i \le g'_i \le g_i + \eps$, we have that the new prices satisfy $p'_i=\E[(X_i - g'_i)^+] \ge \E[(X_i - g_i)^+] - \eps$ and also that $1 \ge \prod_{j < i} F_{j}(g'_i) \ge \prod_{j < i} F_{j}(g_i)$. Therefore, the additive loss of revenue contribution of each alternative $i$ under $\vg'$ is at most $\eps$, thus proving our lemma.

\subsection{Uniform Price Incurs Unbounded Gap (Proof of~\Cref{thm:uniform-cost-gap})}\label{app:unif-prices-gap}

Fix a probability $q\in (0,1)$ and consider an instance $\inst$ composed of $n$ Bernoulli alternatives $X_i\sim A_i\cdot\mathrm{Ber}(q)$ taking value $A_i = (1-(1-q)^i)^{-1}$ with probability $q$ and value $0$ with probability $(1-q)$. Note that the alternatives are already labeled in decreasing order of $A_i = \expectt{X_i|X_i>0}$. 

We first turn our attention to uniform prices, parametrized by a single price $p> 0$. The corresponding indices of the alternatives satisfy $p = \expectt{(X_i-g_i)^+}$ and are therefore given by $g_i = A_i - p/q$; alternatives with $A_i<p/q$ are never inspected by the buyer. Therefore, for a given price $p$, the buyer inspects the alternatives in decreasing order of $A_i$, starting from $X_1$ until the maximum coordinate $k$ for which $A_k\geq p/q$. Since $g_i<A_i$, if an alternative gets realized to the high value the buyer immediately accepts it. If on the other hand it gets realized to $0$, the buyer continues by inspecting the next alternative. Therefore, we have that 
\[\rev(p,\dots,p) = \sum_{i=1}^{k(p)} p\cdot \prod_{j=1}^{i-1}(1-q) = \sum_{i=1}^{k(p)} p\cdot (1-q)^{i-1}\]
where $k(p)=\max\{i:A_i\geq p/q\}$. It is straightforward to see that across all prices $p$ with $k(p)=k$, the one that achieves maximum revenue is $p=A_kq$. Therefore,
\[\max_{p\geq 0}\rev(p,\dots ,p) = \max_{k\in [n]}\sum_{i=1}^{k} A_kq (1-q)^{i-1} = \max_{k\in [n]}A_k(1-(1-q)^k)=1\]

Next, we compute the optimal revenue. By setting $p_i=\mu_i$ for all $i\in [n]$, all the indices become $0$ and so does the buyer's utility (\Cref{fact:utility}). Furthermore, by inspecting the alternatives in decreasing order of $A_i$, we can ensure that the buyer accepts an alternative of value $\max_i X_i$ with probability $1$. Therefore, since these prices simultaneously minimize the buyer's utility and maximize their value, they also maximize the seller's revenue. We conclude that
\[\optrev =\sum_{i=1}^n\mu_i\cdot \prod_{j<i}\probb{X_j=0} = \sum_{i=1}^n A_i q (1-q)^{i-1} = \sum_{i=1}^n \frac{q (1-q)^{i-1}}{1-(1-q)^{i}}=f(q).\]
Therefore, the optimal revenue for $\inst$ is $f(q)$-times larger than the revenue of the optimal uniform-price mechanism. It remains to select the probability $q$.

First, note that $(1-q)^{i-1} \geq 1 - (i-1)q \geq 1-nq\geq 1/2$ for all $q\leq \frac{1}{2n}$. Similarly, we have that $(1-q)^i\geq 1 - iq$ and therefore $1-(1-q)^i \leq iq$. Combining everything, for all $q\leq \frac{1}{2n}$ we have
\[f(q) \geq \sum_{i=1}^n\frac{\frac{1}{2}q}{iq} = \frac{1}{2}\sum_{i=1}^n \frac{1}{i} = \frac{1}{2}H_n\]
and the proof follows.

%% file: appendix_mandatory_ext.tex
\section{Omitted Proofs from~\Cref{sec:mandatory-ext}}\label{app:proofs-for-mandatory-ext}

\subsection{Seller Costs and Outside Options (Proof of~\Cref{lem:4-approx-mandatory-unif-id-sellercost})}\label{app:seller_cost_proof}

Fix any instance $\inst$ and let $\vp$ be the revenue-maximizing prices, with corresponding indices $\vg$. As usual, alternatives are re-labeled so that the buyer inspects them in increasing order of label, that is, $g_1\ge \dots \ge g_n$. We have
\[\optrev = \sum_{i\in [n]: g_i\ge y} (\E[(X_i - g_i)^+]-c_i)\cdot \prod_{j < i} F_{j}(g_i).\]

We first note that $p_i = \E[(X_i - g_i)^+] \ge c_i$ for all $i\in [n]$; otherwise, the corresponding terms are strictly negative and by setting $p'_i=\infty$ for all $i\in S=\{i\in [n]: p_i<c_i\}$ these negative terms will vanish (as the buyer will never inspect) while the probabilities of inspecting all other alternatives $k\notin S$ with $g_k\ge y$ will (weakly) increase from $\prod_{j<k}F_j(g_{k})$ to $\prod_{j<k; j\notin S}F_j(g_{k})$; therefore the total revenue would strictly increase, contradicting the optimality of $\vp$.

Therefore, every term in the optimal revenue is non-negative. Let $\lambda_i=\prod_{j<i}F_j(g_i)$ be the probability of inspecting alternative $i$. Since the indices $g_i$ are non-increasing in $i$, the same will be true for the probabilities $\lambda_i$. Let $k \in [n]$ be the last index where $\lambda_k \ge \frac{1}{2}$, which means $\lambda_{k+1} < \frac{1}{2}$. By replicating the proof of~\Cref{thm:4-approx-mandatory-unif-id}, we can show that the revenue we collected from alternatives $k+1$ through $n$ is at most half the optimal revenue, and that by setting the index of alternatives $1$ through $k$ to $g_k$ their revenue contribution decreases by a factor of at most $2$. This means that the revenue obtained from alternatives $1$ through $k$ under uniform-index $g_k$ is at least $\frac{\optrev}{4}$. Note however that setting the same index to alternatives $k + 1$ through $n$ may introduce net losses to the revenue; in this case we can simply remove the alternatives contributing losses by setting their price to $\infty$.

Now the remaining task is to efficiently compute a price vector that achieves at least the revenue of the above price vector. Using the same argument as \Cref{lem:gittins-index-in-support-mandatory}, we can show that the optimal index $g$ lies in the set $S = \{0\} \cup \bigcup_{i=1}^n \support(X_i)$. By an exchange argument similar to that in \Cref{lem:gitins-index-tiebreak-mandatory}, where the new tie-breaking rule is now by decreasing $\frac{\E[(X_i - g_i)^+] - c_i}{1 - F_i(g_i)}$, we know that we can efficiently compute the expected revenue of any index vector $\vg$. Hence, we can now compute efficiently a price vector that $4$-approximates the optimum.

\subsection{Approximation via Uniform Indices for Uniform Matroids (Proof of \Cref{lem:4-approx-mandatory-unif-id-unifmatroid})}\label{app:uniform-matroid-proof}

Fix any instance $\inst$ corresponding to alternatives $X_1,\dots , X_n$ and let $\matroid$ be the $k$-uniform matroid for some $k\in \{1,\dots, n\}$. Let $\vg = (g_1, g_2, \dots, g_n)$ be the indices induced by any revenue maximizing prices and re-label the alternatives so that the optimal buyer response inspects them in increasing order of label, that is, $g_1 \ge g_2 \ge \dots \ge g_n$. Recall the definition of the random set 
\[
S_i \;:=\; \{\, j<i : \text{$j$ was inspected and } X_j > g_i \,\}.
\]
Since $M$ is the $k$-uniform matroid, we therefore have that $\Pr[i \notin \spn(S_i)] = \Pr[|S_i|<k] = \Pr[|\{j < i : X_j > g_i\}| < k]$. Therefore, we can derive the optimal revenue as
\begin{align*}
    \optrev = \rev(\vg) = \sum_{i=1}^n \E[(X_i - g_i)^+] \cdot \Pr[|\{j < i : X_j > g_i\}| < k]
\end{align*}

Let $\lambda_i := \Pr[|\{j < i : X_j > g_i\}| < k]$ be the probability of inspecting alternative $i$. Since the indices $g_i$ are non-increasing in $i$, the same will be true for the probabilities $\lambda_i$. Let $j \in [n]$ be the last index where $\lambda_j \ge \frac{1}{2}$, which means $\lambda_{j+1} < \frac{1}{2}$. By replicating the proof of~\Cref{thm:4-approx-mandatory-unif-id}, we can show that the revenue we collected from alternatives $j+1$ through $n$ is at most half the optimal revenue, and that by setting the index of alternatives $1$ through $j$ to $g_j$ their revenue contribution decreases by a factor of at most $2$; the arguments are precisely the same and the details are therefore omitted. From this, we obtain that 
\[\rev(g_j) \ge \frac{1}{4}\cdot\optrev.\]

Now the remaining task is to efficiently compute the maximum possible revenue over all uniform-index pricings. Using the exact same arguments as in the proof of \Cref{lem:gittins-index-in-support-mandatory}, we can show that the optimal index $g$ must lie in the set $S=\{0\} \cup \bigcup_{i=1}^n \support(X_i)$ and we can therefore efficiently iterate over all $g\in S$ for discrete distributions (or discretized continuous distributions). Given such a $g\in S$, in order to compute the revenue $\rev(g)$ we only need to compute the revenue-maximizing tie-breaking order and then compute the spanning probabilities under this order (which can also be done efficiently). Using the same exchange argument as in the proof of~\Cref{lem:gitins-index-tiebreak-mandatory}, one can show that the alternatives should still be inspected in decreasing $\E[X_i\mid X_i > g]$, completing the proof.

%% file: appendix_optional.tex
\section{Omitted Proofs from~\Cref{sec:optional}} \label{app:proofs-for-optional}

\subsection{Omitted Proofs for Identical Alternatives Under Optional Inspection}\label{app:proofs-for-optional-iid}

In this section, we present the omitted proofs from~\Cref{sec:optional-iid}, namely~\Cref{thm:optional-iid-2} and~\Cref{thm:optional-iid-approx}.

\begin{proof}[Proof of~\Cref{thm:optional-iid-2}]

We will analytically derive the optimal revenue for $n=2$ identical alternatives under optional inspection and show that the corresponding prices are uniform and efficiently computable. We use $X_1,X_2\sim\dist$ for the values of the alternatives and $\mu$ for their expectation. From now on, let $\vp = (p_1,p_2)$ with $0\leq p_1\leq p_2$ be any set of prices. Clearly, if the buyer inspects at least one alternative, they will first inspect the cheapest one (due to the identical distributions). We will derive a closed-form expression for $\rev(\inst , \vp)$.

Let's condition on the buyer having already inspected the first alternative by paying $p_1$ and having observed a value of $x_1$. At this point, the buyer can either select it (at a reward of $x_1$), or select the second alternative without inspection (at an expected reward of $\mu$) or inspect the second alternative by paying $p_2$ and then accept the best out of the two alternatives. Therefore, the second alternative is only inspected if $-p_2 + \expectt{\max\{x_1,X_2\}}\ge \max\{x_1,\mu\}$, or equivalently,
\[p_2 \leq W(x_1):= \expectt{\max\{x_1,X_2\}} - \max\{x_1,\mu\}.\]
Furthermore, the expected (remaining) utility of the buyer conditioned on $x_1$ will be precisely $\max\{x_1,\mu, -p_2 + \expectt{\max\{x_1,X_2\}}$, or equivalently,
\[
U(x_1):= \max(x_1,\mu) + (W(x_1)-p_2)^+
\]

We now consider the buyer's decision on whether to inspect the first alternative. Once again, the options are to either select any of the two alternatives uninspected (gaining reward $\mu$) or inspect the first one. Therefore, the first alternative is only inspected if $-p_1 + \expectt{U(X_1)} \ge \mu$, or equivalently,
\[p_1 \le \tau(p_2) := \expectt{U(X_1)} - \mu = \expectt{(X_1-\mu)^+} + \expectt{(W(X_1)-p_2)^+} .\]

Combining everything, we obtain that 
\[ \rev(\inst,\vp) =
\begin{cases}
 p_1 + p_2\cdot\probb{p_2\le W(X_1)}  &\text{if } p_1 \le \tau(p_2) \\
 0  &\text{otherwise} \\
\end{cases}
\]
Therefore, to maximize revenue conditioned on the price $p_2$, we should set the first price as high as possible conditioned on $p_1\le \tau(p_2)$ and $p_1\le p_2$, i.e. set $p_1 = \min\{p_2,\tau(p_2)\}$.

We now turn our attention to optimizing over the price $p_2$. Since we have already derived $p_1$ as a function of $p_2$, this is a one-dimensional problem corresponding to maximizing
\[R(p_2)=\min\{p_2,\tau(p_2)\} + p_2 \cdot\probb{p_2\le W(X_1)}.\]

Note that the function $g(p_2):= \tau(p_2)-p_2 = \expectt{(X_1-\mu)^+} + \expectt{(W(X_1)-p_2)^+}-p_2$ is clearly strictly decreasing on $p_2$ and since $g(0) > 0$ and $g(\infty)=-\infty$, it will have a unique root $p^*$ satisfying $p^*=\tau(p^*)$. We distinguish between two different cases:
\begin{itemize}
    
    \item If $p_2 \ge p^*$ then $g(p_2) \le 0$ which implies $\tau(p_2)\le p_2$ and therefore
    \begin{align*}
        R(p_2) &= \tau(p_2) + p_2\cdot\probb{p_2\le W(X_1)} \\
        &= \expectt{(X_1-\mu)^+} + \expectt{(W(X_1)-p_2)^+} + p_2\cdot\probb{p_2\le W(X_1)} \\
        &= \expectt{(X_1-\mu)^+} + \expectt{W(X_1)\cdot\mathbbm{1}[W(X_1)\ge p_2]}
    \end{align*}
    with the last equality following from the identity
    $(W-t)^+ + t\mathbf{1}_{\{W\ge t\}} = W\mathbbm{1}[W\ge t]$. From this expression, we can see that $R(p_2)$ is (weakly) decreasing in $p_2$ and will therefore be maximized at $p_2=p^*$.

    \item If $p_2 \le p^*$ then $g(p_2)\ge 0$ which implies $\tau(p_2)\ge p_2$ and therefore
    \[R(p_2) = p_2 + p_2\cdot\probb{p_2\le W(X_1)}.\]
    Observe that this is precisely the revenue that we would obtain had we set uniform prices $(p_2,p_2)$.
\end{itemize}

Combining everything, we summarize that there exists a uniform pricing $\vp = (p,p)$ with $p\in [0,p^*]$ that achieves optimal revenue. Furthermore, the corresponding revenue as a function of $p$ will be
$R(p) = p(1+\probb{p\le W})$
where $W$ is a random variable corresponding to first sampling $x\sim\dist$ and then returning $W(x) = \expectt{\max\{x,X\}} - \max\{x,\mu\}$. It is not hard to see that the function $R(p)$ is piece-wise linear with break-points corresponding to the support of $W$. Therefore, there exists an optimal price $p$ that lies in the support of $W$; by iterating over all prices $p\in\support(W)$; ignoring any prices above $p^*$ (which we can efficiently compute); and comparing the corresponding revenues $R(p)$, we can efficiently find a price $p$ such that $\rev(\inst, (p,p)) = \optrev(\inst)$.
\end{proof}

\begin{proof}[Proof of~\Cref{thm:optional-iid-approx}] Let $\inst_\optlabel$ denote any \textit{optional inspection} instance over $n\geq 2$ identical alternatives of distribution $\dist$ and let $\mu$ denote the expected value of the alternatives. We consider a \textit{mandatory inspection} instance $\inst_\manlabel^{n-1}$ composed of $(n-1)$-alternatives of distribution $\dist$ and an \textit{outside option} $y=\mu$ that the buyer can accept at any point and halt, at no extra cost. Since the boxes are identically distributed and we are in the mandatory inspection setting, we can use~\Cref{lem:exact-opt-mandatory-special-cases} to efficiently compute\footnote{The proof of~\Cref{lem:exact-opt-mandatory-special-cases} is stated for $y=0$ (i.e. in the absence of an outside option) but the exact same arguments extend to any $y\geq 0$ and the details are therefore omitted.} some $g'\ge y$ such that $\vec{p'}=\vp(g')$ is optimal, i.e. $\optrev(\inst_\manlabel^{n-1}) = \rev(\inst_\manlabel^{n-1},\vec{p'})$. 

Now, consider the price vector $\vp = (\infty ; \vec{p'} )$ for $\inst$ that sets an infinite price for the first alternative and prices alternatives $2$ through $n$. We focus on the behavior of the (optional inspection) buyer under $\vp$. Since the first alternative is priced at $\infty$ and the alternatives are identical (and therefore, they contribute the same expected reward of $\mu$ upon selecting them without inspection), we can assume without loss that (i) the buyer never inspects the first alternative and (ii) if the buyer opts to select an alternative without inspection, it will be the first one. Therefore, the buyer will behave as in the mandatory inspection model over $(n-1)$-alternatives priced at $\vec{p'}$ and an outside option of $y=\mu$ that can be accepted at no cost; i.e., as the buyer in $(\inst_\manlabel^{n-1},\vec{p'})$. From this, we obtain that $\rev(\inst_\optlabel,\vp)=\rev(\inst_\manlabel^{n-1},\vec{p'}) = \optrev(\inst_\manlabel^{n-1})$. Since the prices $\vec{p'}$ can be efficiently computed, the proof will be completed by showing that 
\begin{equation}\label{eq:opt_insp_iid_to_show}
\optrev(\inst_\manlabel^{n-1}) \geq \left(1-\frac{1}{n}\right)\cdot\optrev(\inst_\optlabel).
\end{equation}

To prove this inequality, we consider yet another mandatory inspection instance $\inst_\manlabel^{n}$ comprised by $n$-alternatives of distribution $\dist$ and an outside option of $y=\mu$. Clearly, the optimal revenue we can achieve in $\inst_\manlabel^{n}$ is no less than the optimal revenue in $\inst_\manlabel^{n-1}$ since we have an extra alternative; we first show that it cannot be much larger. In particular, we know that since the alternatives are identically distributed, the optimal revenue in $\inst_\manlabel^{n}$ is also achieved at some uniform index $g''\geq y$. Therefore, we have that
\[\optrev(\inst_\manlabel^{n})=\rev(\inst_\manlabel^{n},g'') = (\expectt{X|X>g''}-g'')\cdot (1-\probb{X\leq g''}^n)\]
where the last equality follows from a straightforward application of~\Cref{claim:revenue-exp-mandatory}. Similarly, we have that 
\[\optrev(\inst_\manlabel^{n-1}) = \rev(\inst_\manlabel^{n-1},g')\geq \rev(\inst_\manlabel^{n-1},g'') = (\expectt{X|X>g''}-g'')\cdot (1-\probb{X\leq g''}^{n-1})\]
and therefore
\begin{equation}\label{eq:opt_insp_iid_add_a_box}
    \frac{\optrev(\inst_\manlabel^{n})}{\optrev(\inst_\manlabel^{n-1})} \leq \frac{1-\probb{X\leq g''}^{n}}{1-\probb{X\leq g''}^{n-1}} \leq \max_{x\in [0,1]}\frac{1-x^n}{1-x^{n-1}} \leq \frac{n}{n-1}
\end{equation}
with the last inequality following by observing that $f(x) = (1-x^n)/(1-x^{n-1})$ is increasing and taking the limit at $x\rightarrow 1$.

Next, we consider the buyer's behavior on instance $\inst_\manlabel^n$ under some set of prices. The key observation is that even if the buyer was in the optional inspection model, their behavior would not change, as there is already an outside option of value $y=\mu$ and therefore we can assume without loss that none of the $n$ alternatives are ever selected without inspection. Therefore, we have $\optrev(\inst_\manlabel^n) = \optrev(\inst_\optlabel')$ where $\inst_\optlabel'$ is an optional inspection instance with $n$ identical alternatives and an outside option of $\mu$. Finally, we will show that $\optrev(\inst_\optlabel)\leq \optrev(\inst_\optlabel')$; once we have that, we also obtain $\optrev(\inst_\optlabel) \leq \optrev(\inst_\manlabel^n)$, and combining with~\eqref{eq:opt_insp_iid_add_a_box}, we immediately obtain~\eqref{eq:opt_insp_iid_to_show} and complete the proof. 

To see why $\optrev(\inst_\optlabel)\leq \optrev(\inst_\optlabel')$, note that both are optional inspection instances over $n$ identical alternatives with outside options $y=0$ and $y=\mu$, respectively. Fix any price vector $\vp$ over the $n$-alternatives and let $p_1\le p_2 \le\dots \le p_n$. We will show that $\rev(\inst_\optlabel,\vp)\leq \rev(\inst'_\optlabel,\vp)$ for all $\vp$. Since the alternatives are identical, the buyer will inspect them in increasing order of label in both instances. We couple the realizations of the alternatives. Focus on any point in time when alternatives $1$ through $j<n$ have been inspected and realized at $x_1,\dots , x_j$. If the buyer decides to not inspect the next box in $\inst_\optlabel'$, it must mean that by accepting $\mu$ and terminating their utility strictly improves. Then, the same will necessarily be true in $\inst_\optlabel$, as the same option of accepting $\mu$ is available (by selecting alternative $j+1$ without inspection) \textit{and} continuing to inspect can only yield as much utility as in $\inst_\optlabel'$ due to the lack of the outside option. Therefore, the probability of inspecting the $j$-th alternative in $\inst_\optlabel$ is upper bounded by the probability of inspecting it in $\inst'_\optlabel$, from which the claim follows.
\end{proof}

\subsection{Redundancy of Reveal Mechanisms under Mandatory Inspection}\label{app:proofs-for-optional-reveal-mandatory}
In this section, we show that in the mandatory inspection setting, the seller never benefits by revealing (and offering) a set of alternatives to the buyer for free, even if the prices of unrevealed alternatives can be adaptive.

\begin{claim}
    Fix any mandatory inspection instance $\inst$. Then, 
    \[\max_{\vp\in\R^n_+}\rev(\inst,\vp) = \max_{\mech \in\revealset}\rev(\inst,\mech).\]
\end{claim}

\begin{proof}
    Since pricing mechanisms correspond to reveal mechanisms with $R=\emptyset$, we immediately have that
    \[\max_{\vp\in\R^n_+}\rev(\inst,\vp) \le \max_{\mech \in\revealset}\rev(\inst,\mech).\]
    In order to prove the inequality in the other direction, we will show that there always exists some set of prices $\vp\in\R^n_+$ such that $\rev(\vp)\ge \max_{\mech \in\revealset}\rev(\mech)$. In particular, let $\mech\in\revealset$ be any maximizer of the right-hand side and let $R\subseteq [n]$ be the set of alternatives it reveals to the buyer. Recall that for any instantiation of $y_R= \max_{i\in R}x_i$, $\mech$ sets a price vector $\vp_{\bar{R}} = \vp_{\bar{R}}(y_R)$ over the remaining items.

    Fix any value $y$ and condition on the event $y_R=y$. Note that after $y_R$ is revealed and given a price vector $\vp$ over $[n]\setminus R$, the buyer's problem now becomes the mandatory inspection setting with an outside option $y_R$ over the unrevealed boxes. Therefore, suppose $\vg_{\bar{R}}$ are the indices of the unrevealed boxes, and that we ordered the boxes so that $[n]\setminus R = [k]$ and the buyer's strategy goes from box $1$ to $k$. Then, the revenue achieved by $\mech$ here is 
    \[\rev(\mech \mid y_R = y) = \sum_{i \in [k]: g_i \ge y} \E[(X_i - g_i)^+] \prod_{j < i} F_j(g_i).\]

    We now define a fixed pricing vector $\vec{p'}$ whose revenue is at least the previous quantity: let $p'_i = p_i$ for $i \in [n]\setminus R$ and $p'_i = \infty$ for $i \in R$. Note that the revenue achieved by this fixed pricing vector is
    \[\rev(\vec{p'}) = \sum_{i \in [k]: g_i \ge 0} \E[(X_i - g_i)^+] \prod_{j < i} F_j(g_i)\]    which is at least the revenue achieved by $\mech$ conditioned on $y_R = y$. Since this quantity is also upper bounded by $\optrev$, we also have $\rev(\mech \mid y_R = y) \le \optrev$. Since this is true for any choice of $y$, we then have $\rev(\mech) \le \optrev$ and the proof follows.
\end{proof}

\subsection{Omitted Proofs for Two-Point Distributions Under Optional Inspection}\label{app:proofs-for-optional-bernoulli}
In this section, we present the omitted proofs of~\Cref{claim:bernoulli_mech_1} and~\Cref{claim:bernoulli_mech_2}. Recall that the $n$ alternatives take values $l_i, r_i > l_i$ with probabilities $(1-q_i)$ and $q_i$ respectively and that alternative $1$ has maximum expectation $\mu_1 \ge \mu_j$ for all $j\in [n]$.

\begin{proof}[Proof of~\Cref{claim:bernoulli_mech_1}]

By definition, the mechanism $\mech_1$ reveals the value of all the alternatives $R=[n]\setminus\{1\}$ except for the first one and then sets the price $p_1 = \expectt{\max\{X_1,y_R\}} - \max\{\mu_1,y_R\}\geq 0$. The buyer can either inspect the first box (and then select the best realization), or select the alternative $j\in R$ of value $x_j=y_R$ for free, or select the first alternative without inspection at an expected reward of $\mu_1$. The price $p_1$ is defined precisely so that the buyer is indifferent towards inspecting the first box and immediately making a selection; since ties are broken in favor of the seller, we therefore have that the first box is always inspected under mechanism $\mech_1$. Taking the expectation over the value of $y_R$, we have
\[
    \rev(\mech_1) = \expectt{\max_{i\in [n]}X_i} - \expectt{\max\{\mu_1, \max_{i\ge 2}X_i\}}.
\]
\end{proof}

\begin{proof}[Proof of~\Cref{claim:bernoulli_mech_2}]

By definition, the mechanism $\mech_2$ is actually a (standard) pricing mechanism corresponding to $p_1 = \infty$ and $p_j$ satisfying $g_j(p_j) = \mu_1$ for all $j\neq 1$. Observe that under this pricing, the buyer will never inspect $X_1$ and will never grab $X_j$ for $j\neq 1$ (as this would provide an expected reward of $\mu_j \leq \mu_1$). Therefore, the buyer will behave exactly as in the mandatory inspection setting over prices $\vec{p'}=(0,p_2,\dots , p_n)$ and alternatives $(X'_1,\dots, X'_n)$ defined as $X'_1 = \mu_1$ and $X'_j = X_j$ for $j\ge 2$. Since all alternatives have the same index $g_i=\mu_1$, from~\Cref{fact:utility}, the utility of the buyer will be precisely 
\[\expectt{\left(\max_{i\in [n]}\min\{X'_i,g_i\}\right)^+}=\mu_1.\]

Furthermore, since all alternatives have the same index $\mu_1$, we can enforce any tie-breaking order. We re-label alternatives $2$ through $n$ so that $r_2\ge r_3\ge \dots \ge r_n$ and assume that the buyer inspects the alternatives $2$ through $n$ in increasing order of label, with alternative $1$ last in the priority. This ensures that whenever the buyer selects an alternative $j\ge 2$, they will select the one of maximum realization. Therefore, the expected value collected by the buyer will be precisely
\[\expectt{\max\{\mu_1, \max_{i\ge 2}X_i\}}.\]
Since revenue is always equal to the value of the buyer minus their utility, we have
\[
    \rev(\mech_2) = \expectt{\max\{\mu_1, \max_{i\ge 2}X_i\}} - \mu_1.
\]
\end{proof}

%% file: appendix_examples.tex
\section{Pandora's Box with Optional Inspection: Some Examples} \label{app:proofs-for-optional-examples}
In this section, we present a series of examples to demonstrate interesting properties for the optional inspection variant of Pandora's Box.

\paragraph{Adaptive Order of Inspection.}
Consider the following $n=3$ alternatives
\[
X_1 =
\begin{cases}
    0, &\text{w.p. }\; 0.75 \\
    1, &\text{w.p. }\; 0.25 \\
\end{cases}
\quad\quad
X_2 =
\begin{cases}
    0, &\text{w.p. }\; 0.75 \\
    2, &\text{w.p. }\; 0.25 \\
\end{cases}
\quad\quad
X_3 =
\begin{cases}
    0, &\text{w.p. }\; 0.5 \\
    \frac{1}{2}, &\text{w.p. }\; 0.5 \\
\end{cases}
\]
and prices $\vp=(0.1,0.25,0)$. We have $\mu_1=0.25$ and $\mu_2=0.5$. Also, $g_1=0.6$ and $g_2=1$.

Since $p_3=0$, the buyer will first inspect this third alternative to learn its realization for free. Let $X_3=y\in\{0,0.5\}$. We will show that if $y=0$ then the buyer proceeds to inspect the first alternative ($X_1$), whereas if $y=0.5$ then the buyer proceeds to inspect the second alternative ($X_2$). This in turn demonstrates that the optimal buyer response in the optional inspection variant of Pandora's Box will depend on past realizations, even in very simple instances. We note that a similar example for three-point distributions is given by~\citet{D18}.

Let's first consider the $y=0$ case. The buyer can either select $y=0$ or one of the two alternatives without inspection (namely the second one, gaining utility $\mu_2=0.5$) or inspect an alternative. 
\begin{enumerate}
    \item If they inspect $X_1$, they pay $0.1$ and observe $X_1\in\{0,1\}$. If $X_1=0$, they are left with $y=0$ and $X_2$ and will therefore select $X_2$ without inspection, gaining utility $\mu_2=0.5$. If $X_1=1$, then they are left with $y=1$ and $X_2$; since $1>\mu_2$ the second alternative will not be selected without inspection and we reduce to the mandatory inspection setting. As $g_2=1$, the buyer's choice between accepting $y$ or $X_2$ is irrelevant and they obtain utility of $1$. Therefore, the buyer's expected utility in that case is 
    \[-0.1 + 0.75\cdot 0.5 + 0.25\cdot 1 = 0.525. \]

     \item If they inspect $X_2$, they pay $0.25$ and observe $X_2\in\{0,2\}$. If $X_2=0$, they are left with $y=0$ and $X_1$ and will therefore select $X_1$ without inspection, gaining utility $\mu_1=0.25$. If $X_2=2$, then they immediately accept $X_2$ since it has the maximum possible realization across all alternatives. Therefore, the buyer's expected utility in that case is 
    \[-0.25 + 0.75\cdot 0.25 + 0.25\cdot 2 = 0.4375. \]
\end{enumerate}
Therefore, for $y=0$, the (uniquely) optimal buyer choice is to inspect alternative $X_1$.

Up next, we consider the $y=0.5$ case. The buyer can either select $y=0.5$ or one of the two alternatives without inspection (again, they prefer the second one) for a total utility of $0.5$ or inspect an alternative. 
\begin{enumerate}
    \item If they inspect $X_1$, they pay $0.1$ and observe $X_1\in\{0,1\}$. If $X_1=0$, they are left with $y=0.5=\mu_2$ and we can therefore assume that $X_2$ is never selected without inspection. Since $g_2=1>0.5$, the second box will be opened and based on whether $X_2\in\{0,2\}$ the buyer gains reward $\{0.5,2\}$ respectively. If $X_1=1$, then the buyer is left with $y=1$ and $X_2$ and collects a future utility of $1$ as we saw before. Therefore, the buyer's expected utility in that case is 
    \[-0.1 + 0.75\cdot (-0.25 + 0.75\cdot 0.5 + 0.25\cdot 2 ) + 0.25\cdot 1 = 0.61875. \]

     \item If they inspect $X_2$, they pay $0.25$ and observe $X_2\in\{0,2\}$. If $X_2=0$, they are left with $y=0.5$ and $X_1$; since $g_1=0.6>0.5$, they will also inspect $X_1$ and collect value $\{0.5,1\}$ based on whether $X_1\in\{0,1\}$ respectively. If $X_2=2$, then they immediately accept $X_2$ since it has the maximum possible realization across all alternatives. Therefore, the buyer's expected utility in that case is 
    \[-0.25 + 0.75\cdot (-0.1 + 0.75\cdot 0.5 + 0.25\cdot 1) + 0.25\cdot 2 = 0.64375. \]
\end{enumerate}
Therefore, for $y=0.5$, the (uniquely) optimal buyer choice is to inspect alternative $X_2$.

\paragraph{Non-Monotonicity of Inspection Probabilities.}

Fix any set of alternatives $\inst=(X_1,\dots , X_n)$ and prices $\vp,\vq\in\R^n_+$ such that $p_1=q_1$ and $p_j \le q_j$ for all $j\ge 2$. In the mandatory inspection setting, since the price of alternatives $j\ge 2$ increases from $\vp$ to $\vq$, the corresponding indices will decrease and it is therefore straightforward to see that probability of the buyer inspecting alternative $1$ under $\vq$ is (weakly) larger compared to the one in $\vp$. Below, we demonstrate that this intuitive monotonicity property is no longer true in the optional inspection setting.

In particular, consider the case where $X_1=2n$ with probability $0.5$ and $X_1=0$ otherwise; and $X_j=n$ with probability $1/n$ and $X_j=0$ otherwise for $j=2,3,\dots , n$. Furthermore, consider the prices $\vp=(\lambda,0,\dots, 0)$ and $\vq=(\lambda,\infty, \dots , \infty)$ for some $\lambda\ge 0$. Since alternatives $2$ through $n$ are offered for free in $\vp$ and cannot be inspected at all under $\vq$, one would expect that the buyer's willingness to inspect the first alternative should be larger in $\vq$. As it turns out, this is not the case.

Under $\vp$, alternatives $2$ through $n$ will be immediately inspected by the buyer for free. The best realization will be $x_{\max}=0$ with probability $(1-1/n)^{n-1}$ and $x_{\max} = n$ otherwise. If $x_{\max}=0$, then the buyer has no incentive to inspect the first alternative as $X_1\ge x_{\max}$ with probability $1$ and will therefore select it without inspection. If $x_{\max}=n$, then the buyer will inspect the first alternative only if  $-\lambda + \expectt{\max\{X_1, n\}} \geq n$, or equivalently, $\lambda\leq n/2$. Thus, we have that
\[\probb{\text{buyer inspects $X_1$ under $\vp$}} = \begin{cases}
     1- (1-\frac{1}{n})^{n-1} & \text{ if } \lambda\leq \frac{n}{2} \\
    0 & \text{ otherwise.}
\end{cases}\]

Under $\vq$, the only alternative that potentially gets inspected is the first one, while boxes $2$ through $n$ can only be selected without inspection at an (expected) reward of $1$. Furthermore, the first alternative can also be selected without inspection at an expected reward of $n$. Therefore, the buyer will only inspect the first alternative if $-\lambda + \expectt{\max\{X_1,1\}} \geq \max\{1,n\}$, or equivalently, $\lambda\leq 1/2$. Thus, we have that
\[\probb{\text{buyer inspects $X_1$ under $\vq$}} = \begin{cases}
     1 & \text{ if } \lambda\leq \frac{1}{2} \\
    0 & \text{ otherwise.}
\end{cases}\]
Therefore, there exists an entire range of prices $p_1=\lambda\in (1/2,n/2]$ for which the probability of inspecting the first alternative under $\vp$ is strictly larger.

\paragraph{Identically Distributed Two-Point Distributions.}

Consider $n$ identically distributed alternatives corresponding to a two-point distribution over values $0\le l\le r$ with probabilities $(1-q)$ and $q$, respectively. Let $\mu = (1-q)l + qr$ be their expectation. Note that in the optional inspection setting, the buyer's utility is lower bounded by $\mu$ and therefore any set of prices for which
\begin{enumerate}
    \item The buyer's utility is $\mu$, and
    \item The buyer always selects an alternative of maximum realization
\end{enumerate}
will maximize the seller's revenue. We will show that setting uniform prices $p=q(1-q)(r-l)$ achieves both of these properties.

First, note that if at any point the buyer observes a value of $r$, they will immediately halt as they can no longer improve the value of their selection. Furthermore, conditioned on reaching the $j$-th alternative and not halting, we know that the first $(j-1)$ alternatives were realized to $l$ and therefore the history doesn't matter; we can treat the remaining alternatives as a fresh instance. 
    
Let $k$ be the number of remaining boxes. If $k=1$, the buyer will immediately select the remaining alternative without inspection, gaining utility $U(1)=\mu$. If $k=2$, the buyer will only inspect the minimum-priced alternative if it is priced at $p$ satisfying 
\[-p + qr + (1-q)U(1)\ge \mu\]
or equivalently, $p\le q(r-\mu) = q(1-q)(r-l)$. Therefore, if both prices are set to this value, the buyer will be indifferent towards inspecting or halting and we have $U(2)=\mu$. We can then repeat the same argument for all $k$.

Thus, if all prices are set at $p=q(1-q)(r-l)$, the buyer will be indifferent towards inspecting an alternative or selecting one without inspection, until a value of $r$ is realized or the last alternative is reached. From this, we obtain that both properties (1) and (2) are satisfied, completing the proof.